\PassOptionsToPackage{table}{xcolor}
\documentclass[]{fairmeta}

\title{TextSeal2: Distortion-Free LLM Watermarking\\ with Copula Mixing and Memoryless Races}

\author[\star]{Tom Sander}
\author[\star]{Pierre Fernandez}
\author{Itai Gat}
\author{Haiming Bao}
\author{Jiaxin Shi}
\author{Jason Park}
\author{Valeriu Lacatusu}
\author{Sylvestre-Alvise Rebuffi}
\author{Hongyan Chang}
\author{Hady Elsahar}
\author{Tom\'{a}\v{s} Sou\v{c}ek}
\author{Tuan Tran}
\author{Surya Parimi}
\author{Alexandre Mourachko}
\affiliation{Meta Superintelligence Labs, \quad $^\star$Equal contributors}

\abstract{
We present \tsii, a distortion-free watermark for large language models that \revdel{is preserved}\rev{keeps its signal} under speculative decoding: tokens that replace rejected drafts remain watermarked with the same keys, and the acceptance rate is unchanged. It also preserves output diversity and adds no measurable inference latency.
\revdel{Previous work introduces diversity into the Gumbel-max watermarking scheme}\rev{Its predecessor, TextSeal, restores diversity in the Gumbel-max scheme} by routing each token to one of two keys via a coin flip. This is suboptimal because the detector averages the scores from both keys without knowing which key was used.
\tsii instead merges \revdel{the }two or more keys into a single virtual key and couples it to fresh noise through AutoMark's Gaussian copula. Any number of keys, including publicly shared keys, can be merged this way without diluting detection.
Under speculative decoding, \revdel{the drafted tokens come from keyed exponential races, and one more exponential race is used for verification.}\rev{\tsii verifies keyed drafts with one more exponential race, run on the leftovers of the draft races.}
\revdel{Because the exponential distribution is memoryless, the distribution of output tokens is that of the target model.}\rev{By memorylessness, this race reproduces block verification, so the output follows the target model exactly.}
We evaluate \tsii on three LLMs. At equal output token diversity, \tsii detects 8 to 22 percentage points more 100-token answers than TextSeal at a false-positive rate of $10^{-3}$, and is comparable to textGrain while requiring no post-logits optimization. 
Under speculative decoding, it detects 18 and 22 percentage points more 200-token answers than TextSeal and textGrain, respectively, at the same acceptance rate.

}

\correspondence{\email{tomsander@meta.com}, \email{pfz@meta.com}}
\metadata[Code]{\url{https://github.com/facebookresearch/textseal}}

\usepackage[utf8]{inputenc}
\usepackage[T1]{fontenc}
\usepackage{url}
\usepackage{amsmath, amssymb, amsthm}
\usepackage{algorithm, algpseudocode}
\usepackage{booktabs}
\usepackage{graphicx}
\usepackage{caption}

\usepackage{xcolor}
\usepackage{colortbl}
\usepackage{tikz}
\usetikzlibrary{arrows.meta,positioning,decorations.pathreplacing,backgrounds,shapes.geometric,fit,calc}
\definecolor{tsred}{HTML}{D1242F}
\usepackage{pifont}
\usepackage{xspace}
\usepackage{enumitem}
\usepackage{needspace}
\usepackage{multirow, makecell}

\newcommand{\cmark}{\textcolor{blue}{\ding{51}}\xspace}
\newcommand{\xmark}{\textcolor{red}{\ding{55}}\xspace}

\DeclareMathOperator*{\argmax}{arg\,max}
\DeclareMathOperator*{\argmin}{arg\,min}
\newcommand{\E}{\mathbb{E}}
\newcommand{\Prob}{\mathbb{P}}
\newcommand{\V}{\mathcal{V}}
\newcommand{\Hnull}{\mathcal{H}_0}
\newcommand{\TV}{\mathrm{TV}}
\newcommand{\Exp}{\mathrm{Exp}}
\newcommand{\Unif}{\mathrm{U}(0,1)}
\newcommand{\sk}{K}
\newcommand{\ts}{TextSeal\xspace}
\newcommand{\tsii}{TextSeal2\xspace}

\newtheorem{theorem}{Theorem}
\newtheorem{proposition}{Proposition}
\newtheorem{lemma}{Lemma}

\theoremstyle{definition}

\newtheorem{remark}{Remark}

\newif\ifshowrev %
\ifshowrev
  \newcommand{\rev}[1]{{\color{teal}#1}}
  \newcommand{\revdel}[1]{{\color{red}[\textbf{removed:} #1]}}
\else
  \newcommand{\rev}[1]{#1}
  \newcommand{\revdel}[1]{}
\fi

\begin{document}
\maketitle

\section{Introduction}
\label{sec:intro}

Identifying AI-generated text helps prevent misuse, from fabricated news~\citep{zellers2019defending,crothers2022machine} to model distillation by malicious actors~\citep{sander2024watermarking}.
An LLM watermark embeds a keyed statistical signal at generation time: each token is sampled from the model's next-token distribution through a coupling with pseudo-random values derived from a secret key and the preceding tokens.
The detector recomputes these values from the key and the text and tests whether the tokens are statistically aligned with them.
The EU AI Act requires providers to mark AI-generated content~\citep{EuropeanAIAct}, and watermarking is deployed at scale: SynthID-Text runs in Gemini~\citep{dathathri2024scalable}, Anthropic marks Claude outputs~\citep{anthropic2026marking}, OpenAI deploys textGrain~\citep{li2026textgrain}, and we released TextSeal~\citep{sander2026textseal}.

We argue that a deployed watermark should meet four requirements: it should 
\textit{(i)} preserve the model's output distribution; \textit{(ii)} keep answers diverse under a fixed watermark key;
\textit{(iii)} add minimal latency to decoding; 
and \textit{(iv)} preserve the acceptance rate of speculative decoding~\citep{leviathan2023fast,chen2023accelerating,sun2025block}.
Schemes whose next-token distribution, averaged over all possible keys, equals the model's are called \emph{distortion-free}.
The earliest such scheme applies the Gumbel-max trick to keyed pseudo-random values~\citep{aaronson2023watermarking}; later ones achieve the same property with other samplers~\citep{christ2023undetectable,hu2023unbiased,wu2023dipmark,kuditipudi2023robust}.
Distortionary schemes such as the green list~\citep{kirchenbauer2023watermark} instead shift token probabilities, trading quality for detectability; we focus on distortion-free schemes.

\begin{table}[t!]
\centering
\footnotesize
\caption{Detection on Qwen3.5-9B (ELI5, temperature 1, top-$p$ 0.95), at fixed diversity: \% of 200-token answers detected at a false-positive rate (FPR) of $10^{-3}$.
Gray: half-width of the 95\% bootstrap interval; bold: best; underlined: second best, per column.
Speculative decoding: Qwen3.5-0.8B drafts 8 tokens per round, with block verification; \tsii uses the residual race, the baselines use block verification's native stopping rule with the last token of each round drawn under an independent key, and textGrain's own recipe verifies the watermarked distributions.
Acceptance is kept when tokens per round do not drop below those of unwatermarked block verification (2.74; paired over 480 prompts, \autoref{fig:ablation} in \autoref{app:details}).
$N_{\mathrm{it}}$: Sinkhorn iterations on a $B \times m$ table (1--3\,ms per token, \autoref{app:latency}); $L$: tournament layers over candidates $\mathcal A$.
\tsii equals AutoMark's copula in distribution; their gap under speculative decoding is the residual race.}
\label{tab:main}
\newcommand{\pmci}[1]{{\color{black!50}\tiny$\pm$#1}}
\setlength{\tabcolsep}{5pt}
\setlength{\aboverulesep}{0pt}
\setlength{\belowrulesep}{0pt}
\renewcommand{\arraystretch}{1.35}
\resizebox{\linewidth}{!}{\begin{tabular}{lw{c}{1.5cm}w{c}{1.5cm}cc}
\toprule
\multirow{2}{*}{Method} & \multicolumn{2}{c}{Speculative decoding} & \multirow{2}{*}{Acceptance kept} & \multirow{2}{*}{Extra post-logits work} \\
\cmidrule(lr){2-3}
& \ding{55} & \ding{51} & & \\
\midrule
SynthID-Text (tournament) & \underline{89}\pmci{3} & 50\pmci{2} & \cmark & $O(L|\mathcal A|)$ \\
TextSeal (coin flip) & 81\pmci{1} & 71\pmci{2} & \cmark & none \\
AutoMark copula & \textbf{91}\pmci{3} & 60\pmci{2} & \cmark & none \\
textGrain & \underline{89}\pmci{4} & 67\pmci{4} & \cmark & Sinkhorn, $O(N_{\mathrm{it}} B m)$ \\
textGrain, own speculative-decoding recipe & -- & \underline{84}\pmci{6} & \xmark{} ($-5\%$ tokens/round) & Sinkhorn, $O(N_{\mathrm{it}} B m)$ \\
\rowcolor{metablue!8}
\textbf{\tsii} (ours) & \textbf{91}\pmci{2} & \textbf{89}\pmci{1} & \cmark & none \\
\bottomrule
\end{tabular}}
\end{table}

\begin{figure}[b!]
\centering
\includegraphics[width=\linewidth]{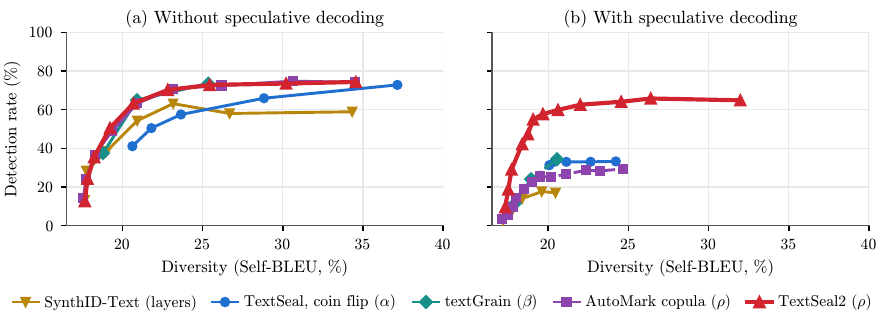}
\caption{
    Diversity--detection trade-off on Qwen3.5-9B at 100 tokens at FPR=$10^{-3}$, (a) without and (b) with speculative decoding (8 drafts per round, block verification). 
    In (a), \tsii coincides with AutoMark's single-key copula as expected since it behaves like a single-key copula. 
    In (b), \tsii uses the residual race, and the baselines use block verification's native stopping rule with the last token of each round drawn under an independent key.
    Diversity is measured by Self-BLEU (lower is more diverse).
}
\label{fig:teaser}
\end{figure}

Every distortion-free watermark couples the sampled token to key-derived randomness~\citep{li2025statistical,xie2025debiasing,long2025optimized,tsur2025heavywater,he2025adaptive,gloaguen2026automark}.
If the key fully determines the token, detection is strongest but repeated generations with the same context never vary; if the token is independent of the key, sampling is unchanged but carries no signal.
Put differently, watermarking removes a fraction of each token's entropy, which is the budget formulation of textGrain~\citep{li2026textgrain}.
Schemes differ in how they trade diversity against detectability.
TextSeal routes each token to one of two keys by a coin flip~\citep{sander2026textseal}; textGrain samples from a regularized optimal-transport coupling solved at every decoding step~\citep{li2026textgrain}; SynthID-Text selects among randomly sampled candidates by tournament~\citep{dathathri2024scalable}; and AutoMark blends its Gumbel values with freshly drawn noise through a Gaussian copula~\citep{gloaguen2026automark}.
Other couplings rely on reweighting~\citep{hu2023unbiased,wu2023dipmark}, inverse-transform sampling~\citep{kuditipudi2023robust}, or maximal couplings~\citep{xie2025debiasing,long2025optimized,tsur2025heavywater}.
Dual-key constructions such as TextSeal also support tampering detection~\citep{feng2026tracing} and public detection, with one key public and the other private~\citep{milis2026public}.
Among these, textGrain detects the most watermarked text at a given diversity, but its optimal transport adds post-logits latency and its speculative variant lowers acceptance, violating our requirements \textit{(iii)} and \textit{(iv)} presented above.

We present \tsii, which retains TextSeal's two keys, meets all four requirements, and matches the state of the art in detectability.
Our starting point is a diagnosis of the gap between TextSeal, textGrain, and AutoMark.
In \ts, each token is generated under one key, but the detector does not know which key, and must average the scores from both keys. 
This halves the signal and reduces the signal-to-noise ratio by $\sqrt 2$.
AutoMark and textGrain do not suffer from this loss because they use a single key, which the detector scores on every token, and draw their diversity from fresh randomness.

\tsii makes three contributions to reach the state of the art (\autoref{tab:main}, \autoref{fig:teaser}):
\begin{itemize}[leftmargin=*,itemsep=1pt,topsep=2pt]
\item \textbf{Copula mixing} (\autoref{subsec:copula}).
We merge the two keys' Gaussian values into a single virtual key and apply AutoMark's copula on it.
Each key retains its role (e.g., one can be published to allow for public detection), and additional keys can be merged without diluting detection.
\item \textbf{The residual race for speculative decoding} (\autoref{subsec:residual}).
Gumbel-max sampling, used for the keyed drafts, is an exponential race.
By memorylessness, the losers' leftovers are fresh exponentials.
We use this property to run one more race on them, and with block verification~\citep{sun2025block}, we decide the first draft token to be rejected, and the bonus token.
The output is exactly distortion-free, acceptance is that of block verification (\autoref{prop:residual}), and detection gains 28 to 33 percentage points over drawing the last token of each round under an independent key.
\item \textbf{A calibrated strength parameter} (\autoref{subsec:entropy}).
The coupling strength $\rho$ plays the role of textGrain's entropy budget $\beta$: the entropy it removes \rev{from a token of positive entropy }increases strictly with $\rho$ (\autoref{thm:monotone}).
No distortion-free scheme can obtain more detection evidence per token than the entropy it removes, and the copula comes close to this bound (\autoref{prop:kl}).
\end{itemize}

\section{Background \& Related work}
\label{sec:background}

\paragraph{Gumbel-max watermarking~\citep{aaronson2023watermarking}.}
At step $t$, the model outputs a distribution $\vec{p}^{(t)} = (p_v^{(t)})_{v \in \V}$ over the vocabulary $\V$, after temperature scaling and top-$p$ truncation (we drop the superscript $t$ when no confusion arises).
A pseudo-random function (PRF) seeded by the secret key $\sk$ and the watermark context window $\mathbf w = (x_{t-k}, \dots, x_{t-1})$ -- the last $k$ tokens -- assigns to each candidate token $v$ a value $R_v = \mathrm{PRF}(v, \mathbf w, \sk) \in [0,1]$, and the watermarked sampler emits:
\begin{equation}
\label{eq:race}
x_t = \argmax_{v \in \V} \; R_v^{1/p_v} = \argmin_{v \in \V} \; \frac{E_v}{p_v}, \qquad E_v = -\ln R_v .
\end{equation}
The right-hand form is an \textit{exponential race}: if the $R_v$ are i.i.d.\ uniform, then the $E_v$ are i.i.d.\ $\Exp(1)$, each $E_v/p_v$ is $\Exp(p_v)$, and the $\argmin$ is distributed exactly according to $\vec p$.
In this race, candidate $v$ runs a lane of length $E_v$ at speed $p_v$, and the first to finish is emitted (\autoref{fig:generation}).
For detection, each observed token $x_i$ receives the score $s_i = -\ln(1 - R_{x_i})$, where $R_{x_i}$ is recomputed from the key and the token's window.
Under the null hypothesis $\Hnull$ (unwatermarked text, or the wrong key), the scores are $\Exp(1)$ and their sum follows a Gamma distribution, giving a simple sum test; watermarked text, whose tokens tend to win their races with unusually large $R$, produces stochastically larger scores.

\paragraph{Repeated windows.}
A sampler is \emph{single-token} distortion-free if, in expectation over the key, every token retains its model probability given the prefix~\citep{dathathri2024scalable}.
It is \emph{single-sequence} distortion-free if the distribution of the entire answer is likewise preserved.
Windows repeat often in code and long reasoning traces.
Reusing a window's keyed values re-runs the same race, so the text is single-token but not single-sequence distortion-free.
Repeated context masking draws fresh values at the later occurrences of a window, i.e.\ plain sampling there~\citep{dathathri2024scalable}.
AutoMark's \texttt{ResidualRace} instead continues the window's race from its leftovers~\citep{gloaguen2026automark}.
Both make the text single-sequence distortion-free (Appendix~\ref{app:proofs-windows}).
At detection, each (window, token) pair is scored once, to avoid biasing the detection test~\citep{fernandez2023three}.

\paragraph{\ts~\citep{sander2026textseal}.}
Plain Gumbel-max is deterministic given the key, so \ts randomises the race \eqref{eq:race} with two keys $k^{(1)}$ and $k^{(2)}$, producing values $R_v^{(1)}$ and $R_v^{(2)}$: at each step it uses $k^{(1)}$ with probability $1-\alpha$ and $k^{(2)}$ with probability $\alpha$.
The detector does not know which key generated a given token, so it fuses the two per-key scores, $s_i = (1-\alpha)\, s_i^{(1)} + \alpha\, s_i^{(2)}$ with $s_i^{(j)} = -\ln(1 - R_i^{(j)})$.
Optionally, each score is weighted by the entropy of the next-token distribution under an auxiliary language model: low-entropy tokens are nearly forced by the context and carry little watermark signal, so down-weighting them improves detection; the weighted sum is then tested against a Gamma distribution with matched mean and variance.
The two keys also serve speculative decoding and public detection.
Under speculative decoding, \ts routes each draft to one of its two keys as usual and generates a replacement token with the other key.
For public detection, $k_{\mathrm{pub}} = k^{(1)}$ is published so anyone can run a weaker test, while $k_{\mathrm{priv}} = k^{(2)}$ stays secret for full verification~\citep{milis2026public}.

\paragraph{Gaussian copula.}
AutoMark~\citep{gloaguen2026automark} randomises differently: it draws independent noise $\eta_v \sim \mathcal N(0,1)$ at every position and replaces $R_v$ by
\begin{equation}
\label{eq:copula}
\tilde R_v = \Phi\big(\rho\, \Phi^{-1}(R_v) + \sqrt{1-\rho^2}\; \eta_v\big),
\end{equation}
where $\Phi$ is the standard normal CDF. The parameter $\rho \in [0,1]$ interpolates between unwatermarked sampling ($\rho = 0$, no dependence on the key) and plain Gumbel-max ($\rho = 1$, deterministic given the key), while $\tilde R_v$ remains uniform -- hence the race distortion-free -- for every $\rho$.

\paragraph{textGrain~\citep{li2026textgrain}.}
At each decoding step, it uses the key and the window to partition the vocabulary into $B$ blocks and to draw a $B \times m$ table of keyed Gumbel costs. Sinkhorn iterations then solve an entropy-regularised optimal-transport problem between the block probabilities and $m$ columns, calibrated so that the key removes a prescribed fraction $\beta$ of the token's entropy; a keyed column is selected, and the block and token are sampled from the resulting coupling with freshly drawn randomness.
Only the key-only computations can overlap the model's forward pass: the Sinkhorn iterations require the next-token distribution, and the following step requires the sampled token, so this work lies on the critical path of decoding. Better kernels can shorten it but not remove it (1--3\,ms per token in our implementation, 0.2--0.4\,ms in a minimal version; \autoref{app:latency}).

\paragraph{Speculative decoding.}
Write $\vec p = P_T(\cdot \mid x_{<t})$ for the target model's distribution and $\vec q = P_D(\cdot \mid x_{<t})$ for the (smaller) draft model's.
In each round, the drafter proposes a block of tokens, and the target evaluates the whole block in one forward pass.
Token-level verification tests the drafts in order.
A draft token $x \sim \vec q$ is accepted with probability $\min(1, p_x/q_x)$, the \emph{raw acceptance test} on the unwatermarked distributions $\vec p$ and $\vec q$.
At the first rejection, the target emits a replacement token from the residual distribution:
\begin{equation}
\label{eq:residual}
r_v = \frac{(p_v - q_v)_+}{\TV(\vec p,\vec q)}, \qquad \TV(\vec p,\vec q) = \sum_v (p_v - q_v)_+ ,
\end{equation}
and discards the later drafts; if every draft is accepted, it emits a bonus token from its next distribution.
The output follows $\vec p$ at acceptance rate $1 - \TV(\vec p, \vec q)$~\citep{leviathan2023fast,chen2023accelerating}.
Block verification~\citep{sun2025block} instead decides on the whole block jointly, so it can keep a longer prefix even when a shorter one fails its per-token test (closed forms in Appendix~\ref{app:proofs-block}).
\citet{sun2025block} prove that repeating such rounds also generates the target's law, and that after any number of rounds block verification has decoded at least as many tokens in expectation as any other verification procedure with this guarantee, so it is never worse than token-level verification.
They report wall-clock speedups of 5--8\% over token-level verification across tasks and datasets.

\section{The \texorpdfstring{\tsii}{TextSeal2} watermark}
\label{sec:method}

\begin{figure}[b!]
\centering
\includegraphics[width=\linewidth,page=1]{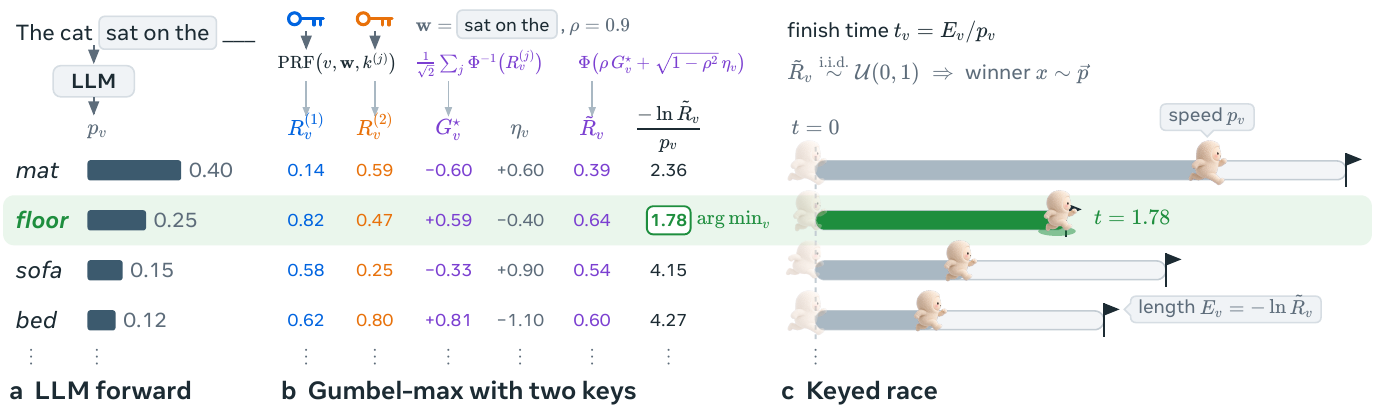}
\caption{\tsii generates one token (\autoref{subsec:copula}).
(a)~The model's forward pass gives the probabilities $p_v$, and the last $k = 3$ tokens form the window $\mathbf w$.
(b)~The keys' values at $\mathbf w$ are merged into the virtual key $G^\star_v$, which the copula \eqref{eq:ts2} blends with the noise $\eta_v$ into $\tilde R_v$; the Gumbel-max choice is $\argmin_v (-\ln \tilde R_v)/p_v$.
(c)~The Gumbel-max choice can be seen as a race \eqref{eq:race}: candidate $v$ runs a lane of length $E_v = -\ln \tilde R_v$ at speed $p_v$, and the token that finishes first is chosen.}
\label{fig:generation}
\end{figure}

\tsii retains \ts's two keys, Gumbel-max sampling \eqref{eq:race}, and detector.
It changes how the keys enter the sampler (\autoref{subsec:copula}), how speculative decoding chooses the tokens it emits (\autoref{subsec:residual}), and how the strength is calibrated (\autoref{subsec:entropy}).
It handles repeated windows as in \autoref{sec:background}, and under speculative decoding every window raced in a round counts as used, whether or not its position is kept (Appendix~\ref{app:proofs-windows}).
\autoref{fig:generation} shows one generation step, \autoref{fig:specdec} one round of speculative decoding, and \autoref{fig:detection} the detection of one text; \autoref{alg:sampler} is the pseudocode.

\subsection{Copula mixing of the two keys}
\label{subsec:copula}

For each key $j \in \lbrace 1, 2\rbrace$, map its values to Gaussian space by $G_v^{(j)} = \Phi^{-1}\big(R_v^{(j)}\big)$, which is standard normal when $R_v^{(j)}$ is uniform.
Instead of running the exponential race on $\vec R^{(1)}$ or $\vec R^{(2)}$, chosen at random as in \ts, \tsii runs the race \eqref{eq:race} on the merged values
\begin{equation}
\label{eq:ts2}
\tilde R_v = \Phi\big(\rho\, G_v^\star + \sqrt{1-\rho^2}\;\eta_v\big), \qquad \text{where } G_v^\star = \frac{G_v^{(1)} + G_v^{(2)}}{\sqrt 2} \text{ and } \eta_v \overset{\text{i.i.d.}}{\sim} \mathcal N(0,1).
\end{equation}
Since the two keys are independent, $G_v^\star$ is standard normal; it depends only on $\left(v, \mathbf w, k^{(1)}, k^{(2)} \right)$ and is independent across candidate tokens $v$ (and across positions whose window has not been used before; see \autoref{sec:background}). Equation~\eqref{eq:ts2} is therefore exactly AutoMark's copula \eqref{eq:copula} applied to a \emph{virtual key} with values $\Phi(G^\star_v)$. In particular, the $\tilde R_v$ are i.i.d.\ uniform and the race is distortion-free for every $\rho$ (\autoref{lem:copula}).
Detection uses the matched score $s_i^\star = -\ln \Phi\big(-G^\star_{x_i}\big)$ -- the usual score $-\ln(1 - R)$ of the virtual key -- which is $\Exp(1)$-distributed under $\Hnull$. The two-key scheme thus has exactly the diversity and detection power of AutoMark's single-key copula at the same $\rho$.
The construction extends to $M$ keys: $G^\star = M^{-1/2}\sum_{k=1}^M G^{(k)}$ remains a single standard normal per candidate, so the full detector loses no signal as keys are added, e.g.\ one publishable key per verifier. (Each key's correlation with the race input, and hence a single-key test's power, shrinks as $\rho/\sqrt M$.)

\paragraph{Why the coin flip loses signal.}
With the coin flip at $\alpha = 1/2$, each token is generated under one key, but the detector does not know which and must score $\frac12 s^{(1)} + \frac12 s^{(2)}$.
Only the key that generated the token carries signal, so averaging halves the expected excess of the score over its null mean of 1.
The two scores are independent under the null, so averaging also halves the variance -- hence the standard deviation falls only by $\sqrt 2$.
The signal-to-noise ratio is therefore reduced by a factor $\sqrt 2$.
Copula mixing has no such loss: every token is generated from $G^\star$, which the detector scores.

\paragraph{Optional public detection.} By default, detection uses both keys through $s_i^\star$. $k_{\mathrm{pub}} = k^{(1)}$ can be published for anyone to compute single-key scores $s_i^{(1)} = -\ln\big(1 - R^{(1)}_{x_i}\big)$. This test is weaker: $G^{(1)}$ has correlation only $\rho/\sqrt 2$ with the race input $\rho G^\star + \sqrt{1-\rho^2}\,\eta$, while $k_{\mathrm{priv}} = k^{(2)}$ supports full-strength verification and exposes manipulations of the public signal~\citep{milis2026public}. SAC-Copula~\citep{li2026saccopula} and AutoMark's Gaussian direction mix keyed Gaussians, without a public key.

\subsection{Speculative decoding through memoryless races}
\label{subsec:residual}

\begin{figure}[b!]
\centering
\includegraphics[width=\linewidth,page=2]{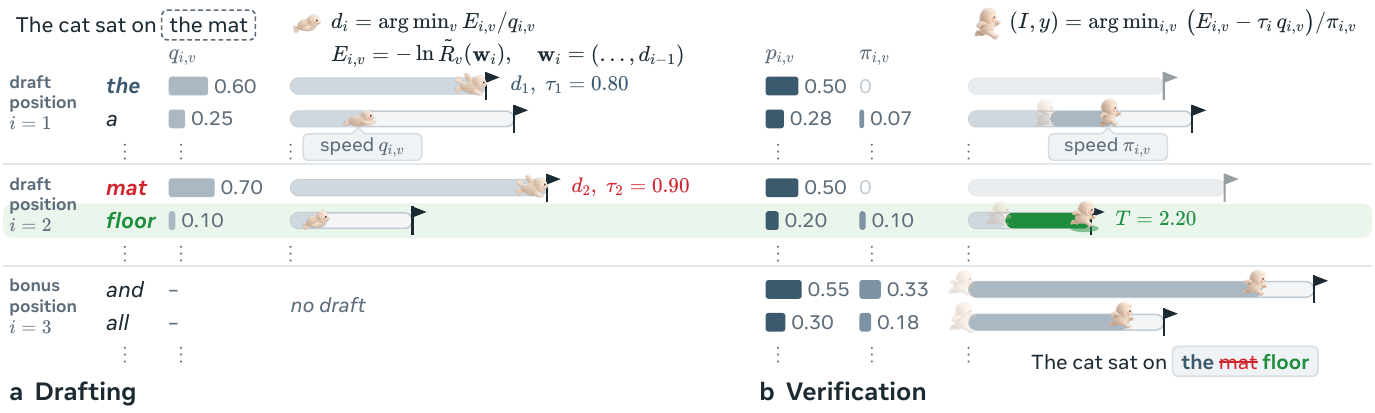}
\caption{Block verification checks the drafts against the target as a block, keeps those before a stop position~$I$, and emits one token $y$ at $I$.
(a)~The drafter proposes ``the mat'' with keyed races.
(b)~One target pass then gives the probability $\pi_{i,v}$ that $I = i$ and $y = v$.
One more race on the leftovers picks $I$ and $y$: ``floor'' wins, so the round keeps ``the'' and emits ``floor'' at $I = 2$.
\tsii preserves the target's distribution and block verification's acceptance rate (\autoref{prop:residual}).
\autoref{alg:sampler} in the appendix gives more details.}
\label{fig:specdec}
\end{figure}

Speculative decoding poses a specific problem for watermarking: the token that verification emits after the kept drafts must follow a residual distribution, such as \eqref{eq:residual}, and must itself carry the watermark.
Prior approaches each give up something.
Re-running the race with weights $\vec r$ on the same keyed values distorts the output, because the draft's win shifts the remaining values (\autoref{rem:naive}).
Replacing from fresh noise~\citep{hu2024inevitable} or from a second key~\citep{sander2026textseal,dathathri2024scalable} keeps the acceptance rate, but the replaced tokens, about a third of the text in our experiments, then lose the watermark or split it across keys.
Verifying between the watermarked distributions keeps the watermark but accepts fewer drafts~\citep{hu2024inevitable,li2026textgrain,dathathri2024scalable} (\autoref{prop:accept}).
\citet{he2026improving} let the detector guess which key made each token, and \citet{liu2026poisson} use several drafts per position.
\tsii keeps both the keys and the acceptance of block verification, with two races on the same keyed values.

\paragraph{Keyed drafts and their leftovers.}
A round has $N$ drafts at positions $i = 1, \dots, N$ and a bonus position $N+1$, with drafter laws $\vec q_i$ and target laws $\vec p_i$ from one target pass.
Position $i$ reads the lane lengths $E_{i,v} = -\ln \tilde R_v(\mathbf w_i)$ from \eqref{eq:ts2} at its window $\mathbf w_i$, which ends with $d_{i-1}$ for $i \ge 2$, and the drafter runs the race at its own speeds:
\begin{equation}
\label{eq:draft}
d_i = \argmin_v \frac{E_{i,v}}{q_{i,v}}, \qquad \tau_i = \frac{E_{i,d_i}}{q_{i,d_i}}, \qquad i = 1, \dots, N.
\end{equation}
When $d_i$ finishes, lane $v$ has covered the distance $\tau_i q_{i,v}$, and $E_{i,v} - \tau_i q_{i,v}$ is its leftover.
By memorylessness, the leftovers are fresh (proof in Appendix~\ref{app:proofs-races}):

\begin{lemma}[Leftovers of an exponential race; AutoMark]
\label{lem:race}
If the $E_v$ are i.i.d.\ $\Exp(1)$\rev{ and $\vec q$ is a probability vector}, the race $x = \argmin_v E_v/q_v$, with time $\tau$, emits $v$ with probability $q_v$, and given $(x, \tau)$ the leftovers $E_v - \tau q_v$, $v \neq x$, are again i.i.d.\ $\Exp(1)$.
\end{lemma}

\paragraph{Block verification and the relay race.}
Given the drafts, block verification (\autoref{sec:background}) stops at a position $I$, keeps $d_{<I} = (d_1, \dots, d_{I-1})$, and emits a token $y$ at position $I$, with law $\pi_{i,v} = \Prob(I = i)\, r_{i,v}$.
Here $\vec r_i$ is a residual law at position $i$, with $\vec r_{N+1} = \vec p_{N+1}$, and $\Prob(I = i)$ depends on all the drafts, since a later success can override an earlier failure; Appendix~\ref{app:proofs-block} gives both in closed form.
\tsii draws this outcome with one more race, on the leftovers of all positions:
\begin{equation}
\label{eq:relay}
(I, y) = \argmin_{(i, v):\; \pi_{i,v} > 0} \; \frac{E_{i,v} - \tau_i\, q_{i,v}}{\pi_{i,v}},
\end{equation}
with $\tau_{N+1} q_{N+1,v} = 0$ at the bonus position, which has no draft race.
Each lane resumes where its draft race left it and runs its leftover at speed $\pi_{i,v}$, as in a relay (\autoref{fig:specdec}), and the round emits $d_{<I}\, y$ without drawing an accept coin.
We call \eqref{eq:relay} the \emph{residual race}: \revdel{``J'' stands for joint, because a single race picks the stopping position and the last token together, and ``residual'' for the leftovers, or residuals, of the draft races on which it runs.}\rev{a single race on the leftovers of the draft races picks the stopping position and the last token together.}
Kept drafts carry the signal of their draft races and the last token that of the relay, so the detector is unchanged: it sums the scores and does not need to locate the ends of rounds.
AutoMark's \texttt{ResidualRace} continues the race of a repeated window, and exponential races\rev{~\citep{li2018poisson}} also couple drafter and target in unwatermarked speculative sampling\revdel{~\citep{kobus2025exponential,daliri2024coupling,li2018poisson}}\rev{~\citep{kobus2025exponential,daliri2024coupling}}; to our knowledge, racing block verification's outcome on the leftovers of keyed drafts is new.

\begin{proposition}[Speculative decoding with the residual race]
\label{prop:residual}
Suppose that every round reads i.i.d.\ uniform keyed values at its windows, independent across positions and of earlier rounds, and independent copula noise.
(i) Given the drafts, $(I, y)$ has law $\vec \pi$, so the number of tokens per round, and hence acceptance, is that of block verification.
(ii) The generated text follows the target's law.
(iii) The window rule (Appendix~\ref{app:proofs-windows}) meets this condition: the text is distortion-free over keys.
\end{proposition}

This proposition says that the residual race preserves the output distribution and the acceptance of block verification.
Appendix~\ref{app:proofs-exact} proves \autoref{prop:residual} in three steps: memorylessness makes the leftovers fresh given the drafts, block verification gives each draft's own lane zero mass (\autoref{lem:zeromass}), and \citet[Theorem~1]{sun2025block} turn block verification's law into the target's.
It also derives the lane length of each emitted token, which is what the detector reads, and shows that racing on the raw lane lengths instead of the leftovers is biased (\autoref{rem:naive}).
Appendix~\ref{app:proofs-accept} proves that no scheme with \revdel{one draft per position}\rev{a single draft} accepts more often while keeping the output distribution (\autoref{prop:accept}), a bound that the residual race attains for $N = 1$, and compares it with the optimality of block verification for $N > 1$.
\autoref{rem:tokenlevel} gives a coin-free token-level variant, and \autoref{alg:sampler} one round in pseudocode.

\subsection{Controlling the entropy removed}
\label{subsec:entropy}

Under a fixed key, the diversity that remains is exactly the entropy the key does not remove: whatever randomness the independent noise still contributes.
textGrain calibrates its strength directly in these terms, removing a prescribed fraction $\beta$ of each token's entropy~\citep{li2026textgrain}.
In \tsii, the coupling strength $\rho$ plays the same role.
\revdel{The}\rev{For a token of positive entropy, the} entropy removed is a strictly increasing function of $\rho$, from none at $\rho = 0$ to the token's full entropy at $\rho = 1$ (\autoref{thm:monotone}), so any target budget can be met by a unique $\rho$, token by token if desired.
In practice this calibration is unnecessary: a fixed $\rho$ removes nearly the same fraction of \revdel{every}\rev{a} token's entropy -- $\rho \approx 0.64$, $0.81$ and $0.93$ remove about 30\%, 50\% and 70\% respectively, across a range of \rev{many-token }distribution shapes (\autoref{app:entropy}).
We therefore use a single constant $\rho$ throughout, with a per-token $\rho_t$ available when an exact budget guarantee is required.

Two results in \autoref{app:optimal} indicate that little is lost by this simplicity.
First, no distortion-free coupling can extract more detection evidence per token, measured as a Kullback--Leibler divergence, than the entropy it removes (\autoref{prop:kl}).
Second, the copula uses this budget near-optimally: on 5{,}169 real Qwen3.5-9B next-token distributions, it attains on average 97\% of the optimal per-token signal-to-noise ratio in the weak-watermark limit\rev{ with a Gaussian score}, and 91--100\% \revdel{at the finite strengths we use}\rev{at $\rho = 0.5$, $0.8$ and $0.95$ on synthetic distributions}.
A per-token optimisation such as textGrain's can therefore gain only a few percent in signal-to-noise ratio, which is consistent with the empirical parity in \autoref{sec:experiments}.

\subsection{Detection and localization}
\label{subsec:detection}

\begin{figure}[b!]
\centering
\includegraphics[width=\linewidth]{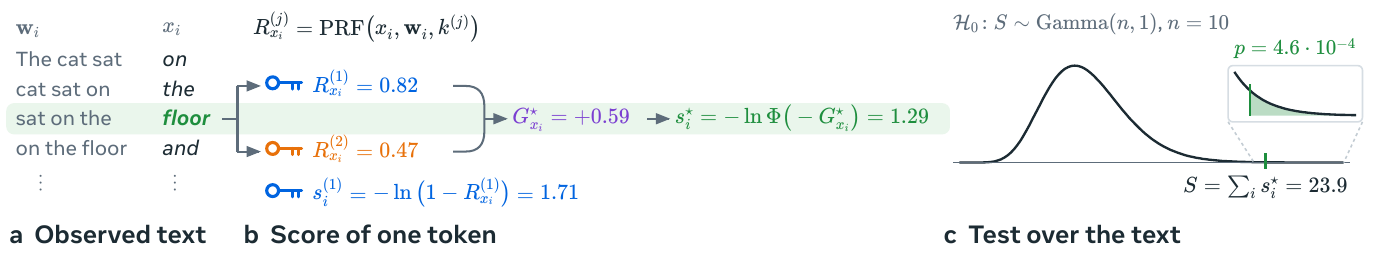}
\caption{\tsii detects the text ``The cat sat on the floor and fell asleep in the sun.'' (\autoref{subsec:detection}).
(a)~Each token $x_i$ and its window $\mathbf w_i$; the highlighted row is the token of \autoref{fig:generation}.
(b)~For ``floor'', the keys' values are merged into $G^\star_{x_i}$ as in \eqref{eq:ts2}, without the noise, and give the score $s^\star_i$; key $k^{(1)}$ alone gives the public score $s^{(1)}_i$.
(c)~The $n = 10$ scores sum to $S = 23.9$, and the $\mathrm{Gamma}(n, 1)$ tail beyond $S$ (box) is $p = 4.6 \cdot 10^{-4}$.}
\label{fig:detection}
\end{figure}

The detector scores each token $x_i$ at its window, with both keys through the virtual key (private score) or with the public key $k^{(1)}$ alone (public score, \autoref{subsec:copula}):
\begin{equation}
\label{eq:scores}
s^\star_i = -\ln \Phi\big(-G^\star_{x_i}\big), \qquad s^{(1)}_i = -\ln\big(1 - R^{(1)}_{x_i}\big).
\end{equation}
Both scores are $\Exp(1)$ under $\Hnull$ for every $\rho_t$, since they do not involve the coupling strength, so the detection machinery of \ts carries over unchanged~\citep{sander2026textseal}.
The sampler can also watermark existing text post hoc, by rephrasing it~\citep{fernandez2025good}.
Keys with windows of different lengths can be merged into $G^\star$ at no cost to the full detector, and because edits break long windows first, the decay of the per-key scores with window length estimates the edit rate (\autoref{app:more}).

\paragraph{Detection tests.}
The private test sums the scores of the $n$ distinct (window, token) pairs, each at its first occurrence (\autoref{sec:background}, \autoref{fig:detection}):
\begin{equation}
\label{eq:sumtest}
S = \sum_{i=1}^n s^\star_i \;\overset{\Hnull}{\sim}\; \mathrm{Gamma}(n, 1), \qquad p = 1 - F_n(S),
\end{equation}
where $F_n$ is the CDF of $\mathrm{Gamma}(n, 1)$; the public test is the same with $s^{(1)}_i$.
Two optional refinements come from \ts~\citep{sander2026textseal}.
\emph{Entropy weighting} down-weights low-entropy tokens, which carry little signal, using the next-token entropy $H_i$ of an auxiliary model:
\begin{equation}
\label{eq:weighted}
S_w = \sum_{i=1}^n w_i\, s^\star_i, \qquad \text{where } w_i = 0.1 + 0.9\, \frac{H_i - H_{\min}}{H_{\max} - H_{\min}}.
\end{equation}
Under $\Hnull$, $S_w$ has mean $\sum_i w_i$ and variance $\sum_i w_i^2$, and the test compares it with the Gamma law of matched moments, of shape $(\sum_i w_i)^2 / \sum_i w_i^2$ and scale $\sum_i w_i^2 / \sum_i w_i$.
The weights do not use the key, so the test is valid for any auxiliary model.
\emph{Localization} finds watermarked spans in documents that mix human and generated text: it tests windows of dyadic lengths $L_0, 2L_0, 4L_0, \dots$, slid at half-length offsets, about $4n/L_0$ candidates, Bonferroni-corrected, and greedily extracts the most significant ones.

\section{Experiments}
\label{sec:experiments}

We compare detection at equal diversity, first without and then with speculative decoding (\autoref{subsec:spec}). Unless stated otherwise, all differences below are in percentage points of detected answers.

\subsection{Setup}

\paragraph{Generation.}
We generate answers to ELI5 questions with Llama-3.1-8B-Instruct, Qwen3.5-9B, and Qwen3.5-27B, at temperature 1 and top-$p$ 0.95, up to 200 new tokens, with watermark window size $k = 3$.
For speculative decoding, Qwen3.5-0.8B drafts $N = 8$ tokens per round for Qwen3.5-9B, and every scheme verifies the drafts by block verification~\citep{sun2025block}; \autoref{app:more} repeats the comparison with $N = 4$ and $16$.
All schemes, \tsii included, follow the repeated-window rules of Appendix~\ref{app:proofs-windows} (one use per window at generation; each (window, token) pair scored once); under speculative decoding, every window that a round reads counts as used, kept or not.

\paragraph{Metrics.}
\emph{Detection} is the fraction of watermarked answers declared watermarked at a nominal false-positive rate of $10^{-3}$. Because no closed-form null is exact in practice, each detector's threshold is the empirical $10^{-3}$ quantile of its p-values on 10{,}000 unwatermarked answers scored under 10 keys ($10^5$ null p-values; calibration in \autoref{fig:calib}). Detection is measured on the first 100 or 200 tokens, with key-only (unweighted) scores unless stated otherwise. We generate 3000 answers per configuration on Qwen3.5-9B, with or without speculative decoding, and 1000 on the other models.
\emph{Diversity} is measured as in \ts, by Self-BLEU among answers to the same question under a fixed key (mean BLEU over all pairs of answers; lower is more diverse), on 100 questions with 32 answers each, or 16 under speculative decoding. Measuring diversity by the share of distinct 4-grams instead leaves the conclusions unchanged, except that textGrain is then slightly ahead of the copulas without speculative decoding (\autoref{app:more}).
To compare at equal diversity, each method's detection is interpolated along its diversity--detection curve at \ts's Self-BLEU; 95\% intervals come from a bootstrap that resamples questions jointly across methods, so intervals for differences are paired.

\begin{figure}[b!]
\centering
\includegraphics[width=0.88\linewidth]{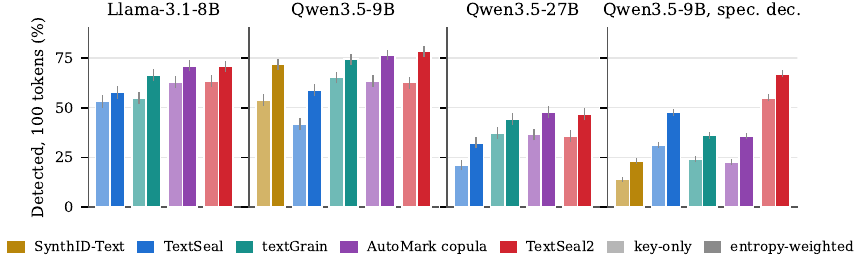}
\caption{Entropy weighting improves detection for every watermark, including under speculative decoding. \% of 100-token answers detected at a false-positive rate of $10^{-3}$ with key-only scores (light bars) and with \ts's entropy-weighted test (full bars), at comparable Self-BLEU. The gains are 5--18 percentage points at 100 tokens and 4--23 at 200 tokens (\autoref{tab:entropy}); without speculative decoding, \tsii and AutoMark's copula coincide in distribution and differ only within their 95\% intervals, as expected.}
\label{fig:entropy}
\end{figure}

\paragraph{Methods.}
We evaluate \ts with $\alpha \in \lbrace 0.5, 0.7, 0.9\rbrace$ (plus 0.8 and 0.95 on Qwen3.5-9B); AutoMark's copula and \tsii with $\rho \in \lbrace 0.7, 0.8, 0.9, 0.95\rbrace$ (on Qwen3.5-9B also 0.5, 0.6, 0.975, 0.99 and 0.995, plus 0.85, 0.925 and 0.999 under speculative decoding), using a 53-bit PRF (\autoref{app:prf}), with \tsii scored by $s^\star$ (or $s^{(1)}$ for the public-key test); our reimplementation of textGrain ($B = m = 256$, $\beta \in \lbrace 0.25, 0.5, 0.75\rbrace$); and, on Qwen3.5-9B, SynthID-Text with 2--30 tournament layers (2--20 under speculative decoding), using its official keyed coins and detection test (fused over its two keys under speculative decoding, \autoref{app:details}).
Under speculative decoding, \tsii uses the residual race (\autoref{subsec:residual}).
The baselines keep block verification's native stopping rule on the unwatermarked laws $\vec p$ and $\vec q$, as \ts and SynthID-Text do with token-level verification~\citep{sander2026textseal,dathathri2024scalable}: each draft comes from the scheme's sampler on $\vec q_i$, and the last token of the round from the same sampler under an independent key on the residual law $\vec r_I$, or from the scheme's normal step on $\vec p_{N+1}$ when it is the bonus token.
\ts routes drafts and bonus tokens between its two keys as usual and uses the other key for the replacement token; textGrain, SynthID-Text and AutoMark's copula, which has no speculative mode of its own, use their first key for drafts and bonus tokens and a second key for replacement tokens.
textGrain's own recipe instead runs block verification between the watermarked distributions, as in its report~\citep{li2026textgrain}.

\subsection{Without speculative decoding}
\label{subsec:nonspec}

\paragraph{The copula dominates the coin flip.}
Interpolated at \ts's Self-BLEU, \tsii detects 10, 22 and 8 percentage points more 100-token answers than the coin flip~\citep{sander2026textseal} on Llama-3.1-8B, Qwen3.5-9B and Qwen3.5-27B respectively (paired 95\% intervals $[+4, +15]$, $[+15, +25]$ and $[+3, +12]$; \autoref{fig:teaser}a and \autoref{fig:three_models}).
Against textGrain, \tsii is ahead on Llama-3.1-8B ($+11$ $[+4, +17]$), level on Qwen3.5-9B ($+1$ $[-5, +8]$), and behind on Qwen3.5-27B ($-5$ $[-10, 0]$).
AutoMark's copula, which coincides with \tsii in distribution, stays within 3 points throughout, confirming that the gain over \ts comes from removing the coin flip rather than from implementation differences.
With the public key alone, \tsii still detects 27\% / 67\% of 100 / 200-token answers at $\rho = 0.9$; merging four keys detects as much as merging two, so extra keys cost nothing (\autoref{app:more}).

\paragraph{Entropy weighting.} \ts's entropy-weighted test improves every method by 4--23 points at 200 tokens, without changing the ranking (\autoref{fig:entropy}, \autoref{tab:entropy}). On Qwen3.5-9B, \tsii detects 78\% / 97\% of 100 / 200-token answers, against 74\% / 97\% for textGrain; under speculative decoding, 67\% / 95\%, against 48\% / 87\% for \ts and 36\% / 76\% for textGrain (Self-BLEU 18.9--20.1).

\paragraph{Robustness and cost.} Under random word substitutions, \tsii retains more detections than \ts and as many as textGrain, while paraphrasing by an LLM removes almost every watermark, as expected for token-level schemes (\autoref{tab:robust}). The gains over the coin flip persist at temperature 0.7 (\autoref{app:more}). \tsii's decoding step is within 0.6\% of \ts's once key-only work overlaps the forward pass (\autoref{tab:ts12}), and its post-logits work is one pass over the vocabulary (\autoref{app:latency}).

\paragraph{Code and long generations.} On the 500 MBPP test problems (\autoref{tab:code}), \tsii detects 55\% of answers at 256 tokens, against 39\% / 46\% / 50\% for \ts / textGrain / SynthID-Text (AutoMark's copula: 59\%). In long thinking traces (up to 2048 tokens; \autoref{tab:repeat}), 16\% of positions on ELI5 and 34\% on MBPP repeat an earlier window. There, continuing each window's race from its leftovers instead of using it once detects $3.8$ $[+2.1, +5.6]$ points more at 256 tokens on ELI5 and $7.6$ $[+2.4, +12.8]$ points more at 512 tokens on MBPP, and raises median z-scores, without measurable distortion (\autoref{tab:repeat}).

\subsection{With speculative decoding}
\label{subsec:spec}

With block verification and 8 drafts, unwatermarked decoding emits 2.74 tokens per round, and every scheme but one keeps this number: paired over 480 prompts, \tsii's residual race differs by $+0.02$ $[-0.01, +0.05]$ tokens per round, as \autoref{prop:residual} predicts, and the baselines by at most $0.04$ (\autoref{fig:ablation}a in \autoref{app:details}).
The exception is textGrain's own recipe, which verifies the watermarked distributions and emits 5\% fewer tokens per round ($-0.14$ $[-0.17, -0.11]$).
Interpolated at \ts's Self-BLEU, \tsii detects 59\% / 89\% of 100 / 200-token answers, against 31\% / 71\% for \ts, 31\% / 67\% for textGrain, 25\% / 60\% for AutoMark's copula, and 17\% / 50\% for SynthID-Text; all paired intervals against \tsii exclude zero (\autoref{tab:main}, \autoref{fig:teaser}b).
textGrain's own recipe comes closest (52\% / 84\%, paired differences $+7$ $[-1, +15]$ / $+5$ $[-1, +11]$ in favour of \tsii), at the cost of acceptance.
With 4 and 16 drafts the ranking is the same, and \tsii leads every baseline that keeps acceptance by 15--41 points at 200 tokens (\autoref{tab:ndrafts}); its lead over textGrain's own recipe is significant with 4 drafts ($+6$ $[+2, +10]$) but not with 16 ($+1$ $[-1, +7]$). A second drafter--target pair from the Llama family gives the same ranking (\autoref{tab:llamapair}). On AIME math problems, with answers of up to 8192 tokens, \tsii also detects the most answers (41\% at 2048 tokens, against 38\% for textGrain's own recipe and 7--22\% for the other baselines), and accuracy is unchanged (\autoref{tab:aime}); sweeping strengths, \tsii also leads at equal diversity, with 30\% detected at 2048 tokens at \ts's Self-BLEU against 23\% for the coin flip and 9--21\% for the other baselines (\autoref{fig:aime}).

\paragraph{Memorylessness is what makes this work.} AutoMark's copula coincides with \tsii in distribution, so the entire gap between them in \autoref{fig:teaser}b (25\% / 60\% against 59\% / 89\%) comes from how the last token of each round is drawn: under an independent key, or by the residual race on the draft races' leftovers. Independent keys waste signal twice over. The last tokens (37\% of the text) carry only the second key's signal, except at the bonus position (5\% of rounds), and under the draft key they score 0.84 on average -- below the null mean of 1, because a token drawn from block verification's residual law lost its draft race -- actively diluting the fused score. The residual race instead keeps the keyed values in play: last tokens score 1.34 under $s^\star$, against 1.42 for kept drafts (\autoref{app:details}).

\paragraph{Distortion.} \autoref{prop:residual} guarantees the exactness of the residual race under an ideal PRF; we check the implementation. Generating with one key per answer (independent draws over keys) and reading the laws $\vec p$ and $\vec q$ used at generation, we test whether $-\ln p_x - H(\vec p)$ over all tokens and $-\ln r_{I,y} - H(\vec r_I)$ at the last token of each round have mean zero, as they must without distortion. \tsii and the baselines are indistinguishable from unwatermarked block verification on both tests (textGrain's own recipe on the all-token test, the only one that applies to it), whereas racing on the raw lane lengths fails the all-token test, by $+11.7$ $[+5.8, +17.8]$ millinats per token (\autoref{fig:ablation}).

\section{Conclusion}
\label{sec:conclusion}

\tsii satisfies the four requirements we set out for a deployed watermark: it is exactly distortion-free, preserves diversity under a fixed key, stays within 0.9\% of plain sampling when key-only work overlaps the forward pass, and maintains the acceptance rate of speculative decoding.
Its detection is comparable to that of textGrain overall -- ahead on Llama-3.1-8B, level on Qwen3.5-9B, and slightly behind on Qwen3.5-27B at 100 tokens -- without textGrain's per-token optimisation.
Our analysis attributes the previous gap to TextSeal's coin-flip routing rather than to the absence of optimal transport: merging the two keys into one virtual key removes the $\sqrt 2$ loss in signal-to-noise ratio, while still allowing one key to be published for public detection.
For speculative decoding, drawing the drafts by keyed races and choosing block verification's outcome by a race on their leftovers carries the watermark without lowering acceptance.

\paragraph{Limitations.}
Our evaluation covers three domains (open-ended ELI5 answers, MBPP code, and AIME math problems under speculative decoding) and three open-weight models, at temperatures 0.7--1.
With block verification and 8 drafts, speculative decoding costs 2--4 percentage points of detection relative to standard decoding, and our speculative-decoding implementation is unbatched, so we do not measure wall-clock speedups.
We do not analyse attacks that exploit knowledge of the public key beyond the private verification it enables~\citep{milis2026public}.
Finally, textGrain's code is not public, so we use our reimplementation (\autoref{app:details}); its deployed parameters may differ.

\clearpage
\bibliographystyle{ieeenat_fullname}
\bibliography{references}

\clearpage
\beginappendix
\section{Proofs}
\label{app:proofs}

\revdel{This appendix proves the results of \autoref{sec:method}: copula mixing (Appendix~\ref{app:proofs-copula}), exponential races (Appendix~\ref{app:proofs-races}), block verification (Appendix~\ref{app:proofs-block}), the exactness of the residual race (Appendix~\ref{app:proofs-exact}), the repeated-window rules on which it relies (Appendix~\ref{app:proofs-windows}), its acceptance (Appendix~\ref{app:proofs-accept}), and the entropy removed (Appendix~\ref{app:proofs-entropy}).
Throughout, keyed values at a previously unused window are i.i.d.\ uniform (ideal PRF), and the copula noise $\eta$ is independent of them.
\tsii draws no other randomness; the accept coins of the token-level scheme in \autoref{rem:tokenlevel} are independent of both.
Results stated in the main text are restated, and each remark follows the result it qualifies.}
\rev{This appendix proves the results of \autoref{sec:method}.
Throughout, keyed values at an unused window are i.i.d.\ uniform (ideal PRF), and the copula noise $\eta$ and the values drawn afresh at used windows, the only other randomness, are independent of them.}

\subsection{Copula mixing}
\label{app:proofs-copula}
\revdel{
This subsection shows that the copula race \eqref{eq:ts2} is distortion-free for every coupling strength, and for any weighting of the two keys.}

\begin{lemma}[Distortion-free for every coupling]
\label{lem:copula}
Let $R^{(1)}_v, R^{(2)}_v$ be i.i.d.\ uniform and $\eta_v$ i.i.d.\ $\mathcal N(0,1)$, independent of them. For every $\rho \in [0,1]$, the values $\tilde R_v = \Phi\big(\rho\, G^\star_v + \sqrt{1-\rho^2}\,\eta_v\big)$ of \eqref{eq:ts2} are i.i.d.\ uniform, so the race \eqref{eq:race} on them emits $v$ with probability $p_v$.
\end{lemma}

\begin{proof}
\revdel{We first follow a single candidate's value through to the race input, then check independence across candidates.

\emph{Step 1: the virtual key is standard normal.}
Each $G^{(j)}_v = \Phi^{-1}(R^{(j)}_v)$ applies the normal quantile function to a uniform variable, so it is $\mathcal N(0,1)$.
The two keys are independent, so their scaled sum $G^\star_v = (G^{(1)}_v + G^{(2)}_v)/\sqrt 2$ is $\mathcal N(0,1)$ too.
Different candidates $v$ use disjoint keyed values, so the $G^\star_v$ are independent across $v$.

\emph{Step 2: the copula output is standard normal.}
The copula forms $Z_v = \rho\, G^\star_v + \sqrt{1-\rho^2}\,\eta_v$ from two independent centred Gaussians.
Its variance is $\rho^2 + (1-\rho^2) = 1$ for every $\rho$, so $Z_v \sim \mathcal N(0,1)$.
Different candidates use disjoint pairs $(G^\star_v, \eta_v)$, so the $Z_v$ are i.i.d.\ $\mathcal N(0,1)$.

\emph{Step 3: the race sees i.i.d.\ uniforms.}
Applying the normal CDF gives $\tilde R_v = \Phi(Z_v)$, uniform on $[0,1]$ and independent across $v$.
The race \eqref{eq:race} on them emits $v$ with probability $p_v$ (\autoref{sec:background}).}
\rev{Each $G^{(j)}_v = \Phi^{-1}(R^{(j)}_v)$ is $\mathcal N(0,1)$, and the two keys are independent, so $G^\star_v$ is $\mathcal N(0,1)$.
$Z_v = \rho\, G^\star_v + \sqrt{1-\rho^2}\,\eta_v$ is a sum of independent centred Gaussians of variance $\rho^2 + (1-\rho^2) = 1$.
Distinct candidates use disjoint keyed values and noise, so the $Z_v$ are i.i.d.\ $\mathcal N(0,1)$, the $\tilde R_v = \Phi(Z_v)$ are i.i.d.\ uniform, and the race emits $v$ with probability $p_v$ (\autoref{sec:background}).}
\end{proof}

\revdel{\begin{remark}[Other ways to combine the keys]
Any weight $a \in [0,1]$ chosen independently of the keyed values (fixed, random, or prefix-dependent) gives a standard normal $G^{(a)}_v = \big(a\, G_v^{(1)} + (1-a)\, G_v^{(2)}\big)/\sqrt{a^2 + (1-a)^2}$, and the same proof applies after conditioning on $a$. \ts's routing is the case $a \in \lbrace 0, 1\rbrace$ drawn at random with $\rho = 1$; \tsii is the balanced case $a = 1/2$.
\end{remark}}

\subsection{Exponential races and their leftovers}
\label{app:proofs-races}

\revdel{This subsection proves \autoref{lem:race} and derives from it the residual race at a single position, which an earlier version of \tsii used to replace rejected drafts.

\noindent\textbf{Lemma~\ref*{lem:race} (restated).} \emph{Let $E_v$ be i.i.d.\ $\Exp(1)$, $q_v > 0$ with $\sum_v q_v = 1$, $\tau = \min_v E_v/q_v$ and $x = \argmin_v E_v/q_v$. Then $\Prob(x = v) = q_v$, and given $(\tau, x)$ the residuals $(E_v - \tau q_v)_{v \neq x}$ are i.i.d.\ $\Exp(1)$.}}
\rev{\noindent\textbf{Lemma~\ref*{lem:race} (restated).} \emph{If the $E_v$ are i.i.d.\ $\Exp(1)$ and $\vec q$ is a probability vector, the race $x = \argmin_v E_v/q_v$, with time $\tau$, emits $v$ with probability $q_v$, and given $(x, \tau)$ the leftovers $E_v - \tau q_v$, $v \neq x$, are again i.i.d.\ $\Exp(1)$.}}

\begin{proof}
Each clock $T_v = E_v/q_v$ rings after an $\Exp(q_v)$ time, with density $q_v e^{-q_v t}$.
The winner $x$ is the first clock to ring, and $\tau$ is its time.
We compute the joint law of the winner, its time, and the losing clocks' \revdel{residuals,} \rev{leftovers,} and check that it factorises.

Fix a token $x$, a time $t \geq 0$, and thresholds $b_v \geq 0$ for $v \neq x$.
The event ``$x$ wins at time $t$ and every other \revdel{residual} \rev{leftover} exceeds $b_v$'' requires $E_x/q_x \in \mathrm dt$ and $E_v > t q_v + b_v$ for $v \neq x$.
The clocks are independent, so this probability is one density term times $|\V|-1$ tail terms:
\[
\Prob\big(x \text{ wins},\ \tau \in \mathrm dt,\ E_v - t q_v > b_v\ \forall v \neq x\big)
= q_x e^{-q_x t}\,\mathrm dt \prod_{v \neq x} e^{-(q_v t + b_v)}
= \underbrace{q_x}_{x} \cdot \underbrace{e^{-t}\,\mathrm dt}_{\tau} \cdot \prod_{v \neq x} \revdel{\underbrace{e^{-b_v}}_{\text{residuals}},} \rev{\underbrace{e^{-b_v}}_{\text{leftovers}},}
\]
using $\sum_v q_v = 1$.

The right side factorises in $x$, $t$, and the $b_v$.
Reading them off: $x$ has law $\vec q$, $\tau$ has density $e^{-t}$ on $t \ge 0$ and is hence $\Exp(1)$, and each \revdel{residual} \rev{leftover} $E_v - \tau q_v$ has tail $e^{-b_v}$, so the \revdel{residuals} \rev{leftovers} are i.i.d.\ $\Exp(1)$, independent of $(\tau, x)$.

Tokens with $q_v = 0$ never win, and their $E_v$ are left untouched.
\end{proof}

\revdel{At a single position, let $\vec r$ be a law that excludes the draft $x$, as the residual \eqref{eq:residual} does after a rejection ($r_x = 0$).
Given $(x, \tau)$, the leftovers divided by $r_v$ are independent $\Exp(r_v)$ variables by \autoref{lem:race}, so the race
\begin{equation}
\label{eq:replace}
y = \argmin_{v:\, r_v > 0} \; \frac{E_v - \tau\, q_v}{r_v}
\end{equation}
samples exactly from $\vec r$, while reusing the keyed values that drew the draft; tokens with $q_v = 0$ keep $E_v$ as their leftover.
In uniform form, \eqref{eq:replace} is the race $\argmax_v (R'_v)^{1/r_v}$ on the rescaled values $R'_v = \tilde R_v / \tilde R_x^{\,q_v/q_x}$.}

\subsection{Block verification}
\label{app:proofs-block}

\revdel{This subsection gives the outcome law of block verification in closed form and shows that it puts no mass on the drafts' own lanes.}
Fix a round with drafts $d_1, \dots, d_N$, drafter laws $\vec q_1, \dots, \vec q_N$ and target laws $\vec p_1, \dots, \vec p_{N+1}$.
For $i \le N$, block verification~\citep{sun2025block} uses
\begin{equation}
\label{eq:block}
c_0 = 1, \quad c_i = \min\Big(1,\ c_{i-1}\, \frac{p_{i,d_i}}{q_{i,d_i}}\Big), \quad u_{i,v} = (c_{i-1}\, p_{i,v} - q_{i,v})_+, \quad z_i = \sum_v u_{i,v}, \quad r_{i,v} = \frac{u_{i,v}}{z_i},
\end{equation}
with $\vec r_{N+1} = \vec p_{N+1}$.
By the retention identity of \citet[Lemma~3]{sun2025block}, $c_i$ is the probability that the round keeps $d_1, \dots, d_i$, given \revdel{these drafts}\rev{its first $i$ drafts and averaged over the later ones}.
It drops when the target likes a draft less than the drafter does and recovers when later drafts are likelier under the target, so a likely later draft can rescue an unlikely earlier one.
With $c_{i-1} = 1$, $\vec r_i$ is the residual \eqref{eq:residual} of token-level verification.
The success probabilities are
\[
h_1 = 1, \qquad h_i = \frac{z_i}{z_i + 1 - c_{i-1}} \quad (2 \le i \le N), \qquad h_{N+1} = c_N,
\]
with $h_i = 1$ when $z_i$ and $1 - c_{i-1}$ both vanish.
Native block verification tests each position $i \ge 2$ with an independent uniform, which succeeds with probability $h_i$, stops at the last success, or at position 1 if every test fails, and draws $y \sim \vec r_I$.
Given the drafts, its outcome therefore has the law
\begin{equation}
\label{eq:stoplaw}
\pi_{i,v} = \Prob(I = i)\, r_{i,v}, \qquad \Prob(I = i) = h_i \prod_{j > i} (1 - h_j),
\end{equation}
whose total is $1 - \prod_{j=1}^{N+1} (1 - h_j) = 1$, since $h_1 = 1$.
\citet{sun2025block} index by the number of kept drafts: their $p_i$, $h^{\mathrm{block}}_i$ and $\tau$ are our $c_i$, $h_{i+1}$ and $I - 1$, and their residual after $\tau$ kept drafts is $\vec r_{\tau+1}$.
\revdel{For $N = 1$, $\Prob(I = 1) = 1 - \min(1, p_{1,d_1}/q_{1,d_1})$ and $\vec r_1$ is the residual \eqref{eq:residual}, so block verification is token-level verification.}
\rev{For $N = 1$ it is token-level verification: $\Prob(I = 1) = 1 - \min(1, p_{1,d_1}/q_{1,d_1})$, and $\vec r_1$ is the residual \eqref{eq:residual}.}

\begin{lemma}[Zero mass on the draft lanes]
\label{lem:zeromass}
For every $i \le N$, $\pi_{i,d_i} = 0$, and $\Prob(I = i) = 0$ if $z_i = 0$.
\end{lemma}

\begin{proof}
If $c_i = 1$, then $h_{i+1} = 1$: for $i < N$ because $1 - c_i = 0$, and for $i = N$ because $h_{N+1} = c_N$.
The factor $1 - h_{i+1}$ of $\Prob(I = i)$ then vanishes.
If $r_{i,d_i} > 0$, then $c_{i-1} p_{i,d_i} > q_{i,d_i}$, so $c_i = 1$ and $\Prob(I = i) = 0$; otherwise $r_{i,d_i} = 0$.
Either way, $\pi_{i,d_i} = 0$.
If $z_i = 0$ and $c_{i-1} < 1$, then $i \ge 2$ and $h_i = 0$.
If $z_i = 0$ and $c_{i-1} = 1$, then $p_{i,v} \le q_{i,v}$ for every $v$, so $\vec p_i = \vec q_i$, $c_i = 1$, and $\Prob(I = i) = 0$ as above.
\end{proof}
So $\vec r_i$ is needed only when $z_i > 0$.

\subsection{Exactness of the residual race}
\label{app:proofs-exact}

\revdel{This subsection proves \autoref{prop:residual}, derives what the detector reads, and qualifies the result with two remarks.
The proof has three steps.
First, each draft race reads only the values of its own position, and positions have independent values, so given the drafts the leftovers of the losing lanes are i.i.d.\ $\Exp(1)$ by \autoref{lem:race}, as are the values of the bonus position.
Second, the only other lane is the draft's own, whose leftover is 0, and \autoref{lem:zeromass} gives it zero mass: if $r_{i,d_i} > 0$, then $c_{i-1} p_{i,d_i} > q_{i,d_i}$, so $c_i = 1$, and the round keeps $d_i$ for sure and never stops at position $i$.
Third, the relay is then an exponential race over fresh values with weights $\vec \pi$, so $(I, y)$ follows $\vec \pi$, and \citet[Theorem~1]{sun2025block} turn block verification's law into the target's.
}
\noindent\textbf{Proposition~\ref*{prop:residual} (restated).} \emph{Suppose that every round reads i.i.d.\ uniform keyed values at its windows, independent across positions and of earlier rounds, and independent copula noise.
(i) Given the drafts, $(I, y)$ has law $\vec \pi$, so the number of tokens per round, and hence acceptance, is that of block verification.
(ii) The generated text follows the target's law.
(iii) The window rule \revdel{(\autoref{sec:method}, Appendix~\ref{app:proofs-windows})} \rev{(Appendix~\ref{app:proofs-windows})} meets this \revdel{condition, so} \rev{condition:} the \revdel{generated} text is \revdel{exactly distortion-free, averaged} \rev{distortion-free} over keys.}

\begin{proof}
Fix everything before the round, and write $\vec E_i = (E_{i,v})_{v \in \V}$ for the lane lengths of position $i$.

\emph{(i) Leftovers given the drafts.}
Fix a sequence $\delta = (\delta_1, \dots, \delta_N)$ of drafts.
The windows $\mathbf w_i$ and the laws $\vec q_i$, $\vec p_i$ and $\vec \pi$ are then fixed, and by assumption and \autoref{lem:copula}, $\vec E_1, \dots, \vec E_{N+1}$ are independent families of i.i.d.\ $\Exp(1)$ variables.
The event $\lbrace d = \delta \rbrace$ is the intersection of the events $\lbrace \argmin_v E_{i,v}/q_{i,v} = \delta_i \rbrace$, $i \le N$, each of which involves only $\vec E_i$.
So given $d = \delta$, the families stay independent, $\vec E_{N+1}$ keeps its law, and each $\vec E_i$, $i \le N$, has the law of i.i.d.\ $\Exp(1)$ variables conditioned on $\delta_i$ winning the race with weights $\vec q_i$.
By \autoref{lem:race}, the leftovers $E_{i,v} - \tau_i q_{i,v}$, $v \neq \delta_i$, are then i.i.d.\ $\Exp(1)$ and independent of $\tau_i$; tokens with $q_{i,v} = 0$ keep $E_{i,v}$ as their leftover.
Hence, given the drafts, the leftovers of all lanes but the drafts' own, together with $\vec E_{N+1}$, are i.i.d.\ $\Exp(1)$ and independent of $\tau_1, \dots, \tau_N$.

\emph{(i) Relay.}
A draft's own lane has leftover $E_{i,d_i} - \tau_i q_{i,d_i} = 0$, but $\pi_{i,d_i} = 0$ by \autoref{lem:zeromass}, so \eqref{eq:relay} reads only the leftovers above.
It is then the race of \autoref{lem:race} over the outcomes with $\pi_{i,v} > 0$, with these leftovers as values and the weights $\vec \pi$, which sum to 1.
So $\Prob(I = i, y = v \mid d) = \pi_{i,v}$, the law of block verification given the drafts, and the proof of \autoref{lem:race} shows that the winning time $T$ is $\Exp(1)$ and independent of the winner.
By \autoref{lem:race}, $d_i \sim \vec q_i$ given $d_{<i}$, so $(d, I, y)$ has the law of a round of native block verification; in particular, \revdel{the number $I$ of tokens per round has its law, and} acceptance is kept.

\emph{(ii) Output.}
Given everything before it, each round emits $d_{<I}\, y$ with the law of a native round, so by induction over rounds the text has the law of speculative decoding with block verification, which is the target's law~\citep[Theorem~1]{sun2025block}.

\emph{(iii) Window rule.}
Let $\mathcal F$ contain the emitted prefix and all randomness used to generate it.
Under the rule of Appendix~\ref{app:proofs-windows}, position $i$ of a round reads either the keyed values of a window that was raced neither before the round nor at an earlier position of the round, or freshly drawn values.
Keyed values at a window that no race has read are i.i.d.\ uniform and independent of the values at other windows and of $\mathcal F$, and the copula noise is drawn anew at every position.
So \revdel{for every sequence of drafts,} the values read at positions $1, \dots, N+1$ are \revdel{independent,} i.i.d.\ \revdel{uniform,} \rev{uniform} and independent of $\mathcal F$, \revdel{which is all that} \rev{as} (i) and (ii) \revdel{use.} \rev{require.}
\end{proof}

\paragraph{What the detector reads.}
A kept draft $d_i$, $i < I$, has lane length $E_{i,d_i} = \tau_i q_{i,d_i}$, where $\tau_i$ is $\Exp(1)$ and independent of $(I, y)$ given the drafts, because the relay reads only leftovers.
Kept drafts therefore carry the signal of their draft races.
The last token has lane length $E_{I,y} = \tau_I q_{I,y} + \pi_{I,y}\, T$, with $T \sim \Exp(1)$ the winning time of the relay, and carries the signal of the relay, which includes the choice of where to stop.
\revdel{Every emitted token is thus selected by a race on its own keyed values, so the detector is the one used without speculative decoding: the plain sum of scores, which does not need to locate the ends of rounds.

In uniform form, \eqref{eq:relay} is the race $\argmax_{(i,v)} (R'_{i,v})^{1/\pi_{i,v}}$ on the rescaled values $R'_{i,v} = \tilde R_v(\mathbf w_i) / \tilde R_{d_i}(\mathbf w_i)^{q_{i,v}/q_{i,d_i}}$ for $i \le N$ and $R'_{N+1,v} = \tilde R_v(\mathbf w_{N+1})$, which extends the rescaling of \eqref{eq:replace} to all positions.}
\rev{Every emitted token is selected by a race on its own keyed values, so the detector is unchanged: it sums the scores and needs no round boundaries.}

\begin{remark}[Racing on raw lane lengths is biased]
\label{rem:naive}
Racing on the raw lane lengths $E_{i,v}$ instead of the leftovers uses $E_{i,v} = \tau_i q_{i,v} + (E_{i,v} - \tau_i q_{i,v})$.
The extra term $\tau_i q_{i,v}$ penalises the tokens the drafter liked and the positions whose draft race ended late, while the bonus position has no such term.
With one draft, $\vec q_1 = (0.6, 0.3, 0.1)$, $\vec p_1 = (0.2, 0.4, 0.4)$ and $d_1$ the first token, the residual race stops at position 1 with probability $2/3$ and then emits the third token with probability $r_{1,3} = 3/4$.
On the raw lane lengths, the race stops there with probability $85/144 \approx 0.59$ and emits the third token with probability $15/17 \approx 0.88$.
The same bias affects the single-position rule that re-runs the race with weights $\vec r_1$ on the raw values after a rejection by a fresh coin: it picks the third token 86\% of the time.
\end{remark}

\begin{remark}[Token-level verification]
\label{rem:tokenlevel}
Token-level verification keeps draft $i$ with probability $a_i = \min(1, p_{i,d_i}/q_{i,d_i})$ and, at the first rejection, emits a token from the residual $\vec r^{\,\mathrm{tok}}_i$ of \eqref{eq:residual} at that position.
Given the drafts, its outcome has the law
\[
\pi^{\mathrm{tok}}_{i,v} = \Big(\prod_{j < i} a_j\Big) (1 - a_i)\, r^{\mathrm{tok}}_{i,v} \quad (i \le N), \qquad \pi^{\mathrm{tok}}_{N+1,v} = \Big(\prod_{j \le N} a_j\Big)\, p_{N+1,v}.
\]
If $a_i = 1$, the factor $1 - a_i$ vanishes; otherwise $p_{i,d_i} < q_{i,d_i}$ and $r^{\mathrm{tok}}_{i,d_i} = 0$.
Either way, this law also vanishes on the draft lanes.
The proof of \autoref{prop:residual} then applies with $\vec \pi^{\mathrm{tok}}$ in place of $\vec \pi$, and with the validity of token-level verification~\citep{leviathan2023fast,chen2023accelerating} in place of that of block verification: the race \eqref{eq:relay} on $\vec \pi^{\mathrm{tok}}$ keeps token-level verification's law without a coin.
\revdel{An earlier version of \tsii instead keeps draft $i$ if a freshly drawn $C_i \sim \Unif$ satisfies $C_i < a_i$, replaces the first rejected draft by \eqref{eq:replace} at its position, and draws the bonus token by the race \eqref{eq:race} on $\vec p_{N+1}$.
Rejection depends only on $(d_i, C_i)$, so by \autoref{lem:race} the replacement follows $\vec r^{\,\mathrm{tok}}_i$.
The emitted token is then $v$ as a kept draft with probability $\min(p_{i,v}, q_{i,v})$ and as a replacement with probability $(p_{i,v} - q_{i,v})_+$, which total $p_{i,v}$, at acceptance rate $1 - \TV(\vec p_i, \vec q_i)$.
In that scheme, the drafts after a rejection influence no emitted token, so a draft and its replacement count as one use of a window and discarded drafts count as none.}
\end{remark}

\subsection{Repeated windows}
\label{app:proofs-windows}
\revdel{
This subsection states the repeated-window rules of \autoref{sec:background} and \autoref{sec:method} and shows that they keep the generated text single-sequence distortion-free and the detection tests of \autoref{subsec:detection} valid.}

\paragraph{The rules.}
At generation, a key's values at a given window serve at most one emitted token per answer, and a position whose window was already used draws its values afresh, which is plain sampling (\autoref{app:details}).
Under speculative decoding, a window counts as used once a round races it, whether or not its position is kept.
The positions after $I$ are discarded, but their keyed values still shaped the outcome: the relay reads the leftovers of every position, and $\Prob(I = i)$ depends on the later drafts.
A position whose window was already raced, earlier in the text or earlier in the same round, therefore uses freshly drawn values, so the positions of a round have independent values.
A round races $N+1$ windows but emits only $I$ tokens, so more positions fall back to fresh values than without speculative decoding.
Without speculative decoding, \tsii can instead continue a window's race at each later visit, as the \texttt{ResidualRace} of AutoMark does~\citep{gloaguen2026automark}: it keeps the leftovers of the window's last race and gives the token that won it a freshly drawn exponential.
At detection, each (window, token) pair is scored once, at its first occurrence.

\paragraph{The whole text is distortion-free.}
Reusing a window's keyed values re-runs the same race: each token still follows the model given its prefix, averaged over keys, but tokens at repeated windows become dependent, so the law of the whole answer is not the model's.
Under the rules, each position reads keyed values that no earlier race has read, or freshly drawn values, which are independent of the prefix and of the randomness that produced it.
Each token then follows the model given its prefix, and the chain rule over positions gives the model's law for the whole answer; under speculative decoding, this is part (iii) of \autoref{prop:residual}, whose proof uses exactly this independence.
With leftover continuation, given the earlier visits of a window, the leftovers of its losing tokens are again i.i.d.\ $\Exp(1)$ by \autoref{lem:race}, and the winner's fresh value is independent of them, so every visit is distortion-free given the \revdel{past, jointly as well.} \rev{past.}

\paragraph{The detection tests stay valid.}
Under $\Hnull$, the text does not depend on the key, so with an ideal PRF the scores \eqref{eq:scores} of distinct (window, token) pairs are i.i.d.\ $\Exp(1)$ given the text, while a repeated pair only repeats its score.
Scoring each pair once leaves $n$ distinct pairs, so $S$ in \eqref{eq:sumtest} follows $\mathrm{Gamma}(n, 1)$ exactly, and distinct tokens at the same window still give independent scores.
The score of a pair depends only on the key, the window and the token, so it is the same at every occurrence, under either generation rule.
When cropping or editing changes which occurrence of a pair comes first, the detector therefore reads the same score for that pair.

\subsection{Acceptance}
\label{app:proofs-accept}

\revdel{This subsection shows that the residual race keeps the best acceptance available, with one draft and with several.}
For $N = 1$, it keeps the draft with probability $1 - \TV(\vec p_1, \vec q_1)$, and no scheme with a single draft at a position does better while preserving the output distribution.

\begin{proposition}[Acceptance bound]
\label{prop:accept}
At one position, let a scheme draft $D$ and emit $O$, and let $\tilde q$ and $\tilde p$ be their laws given the key, with $\E[\tilde q] = \vec q$ and $\E[\tilde p] = \vec p$ over keys.
Then $\Prob(\mathrm{accept}) \le 1 - \E[\TV(\tilde p, \tilde q)] \le 1 - \TV(\vec p, \vec q)$, and the second inequality is strict iff for some $v$, $\tilde p_v - \tilde q_v$ takes both signs with positive probability.
Verifying on watermarked laws (accepting $D$ with probability $\min(1, \tilde p_D/\tilde q_D)$) attains the first bound; the raw test with an independent coin, and the residual race with $N = 1$, attain the second.
\end{proposition}

\begin{proof}
We adapt the sum-over-min step and the expected-TV step from the proof of Theorem 1 of \citet{hu2024inevitable}, with the second step proved there through the dual form of TV.
Fix the key, and take probabilities given it.
An accepted draft is the emitted token, and $(D, O)$ is a coupling of $\tilde q$ and $\tilde p$, so
\[
\Prob(\mathrm{accept}) \le \Prob(O = D) = \sum_v \Prob(D = O = v) \le \sum_v \min(\tilde p_v, \tilde q_v) = 1 - \TV(\tilde p, \tilde q).
\]
Averaging over keys gives the first inequality.

For the second, let $Y_v = \tilde p_v - \tilde q_v$, so that $\TV(\tilde p, \tilde q) = \sum_v (Y_v)_+$ and $\E[Y_v] = p_v - q_v$.
Since $y \mapsto y_+$ is convex, Jensen's inequality gives $\E[(Y_v)_+] \ge (p_v - q_v)_+$ for each $v$.
Summing over $v$ gives the second inequality, which is strict iff one of these terms is.
If $Y_v$ takes both signs with positive probability, then $(Y_v)_+ > Y_v$ on $\lbrace Y_v < 0\rbrace$ and $(Y_v)_+ > 0$ on $\lbrace Y_v > 0\rbrace$, so $\E[(Y_v)_+] > \max(\E[Y_v], 0)$ and the term is strict.
Otherwise $(Y_v)_+$ is almost surely $Y_v$ or almost surely $0$, and the term is an equality.

At a fixed key, verifying on the watermarked laws is speculative sampling with draft law $\tilde q$ and target law $\tilde p$, which accepts with probability $\sum_v \tilde q_v \min(1, \tilde p_v/\tilde q_v) = 1 - \TV(\tilde p, \tilde q)$.
The raw accept test with an independent coin accepts with probability $\sum_v q_v \min(1, p_v/q_v) = 1 - \TV(\vec p, \vec q)$, since the draft follows $\vec q$ over keys.
Given the draft, the residual race with $N = 1$ keeps it with probability $\Prob(I = 2) = c_1 = \min(1, p_D/q_D)$ by \autoref{prop:residual}, the same as the raw test.
\end{proof}

For $N > 1$, block verification decodes at least as many tokens in expectation as any \revdel{verification} \rev{valid draft verification} \revdel{procedure that preserves the target's law~\citep[Theorem~2]{sun2025block},} \rev{algorithm~\citep[Theorem~2]{sun2025block},} and the residual race \revdel{has} \rev{inherits} \revdel{the} \rev{it} \revdel{same acceptance by \autoref{prop:residual}.} \rev{(\autoref{prop:residual}).}

\subsection{Entropy removed}
\label{app:proofs-entropy}
\revdel{
This subsection proves \autoref{thm:monotone}: the entropy that the key removes is continuous and strictly increasing in $\rho$.}

Let $\pi_\rho(\cdot \mid G^\star)$ be the law of the emitted token $X$ given the keyed Gaussians, over the newly drawn noise. As in textGrain~\citep{li2026textgrain}, the entropy the key removes is
\begin{equation}
\label{eq:delta}
\Delta(\rho; \vec p) = I(G^\star; X) = H(\vec p) - \E_{G^\star}\Big[H\big(\pi_\rho(\cdot \mid G^\star)\big)\Big].
\end{equation}

\begin{theorem}[The coupling sets the entropy removed]
\label{thm:monotone}
\rev{If $H(\vec p) > 0$, then }$\rho \mapsto \Delta(\rho; \vec p)$ is continuous and strictly increasing on $[0,1]$, from $\Delta(0;\vec p) = 0$ to $\Delta(1;\vec p) = H(\vec p)$.
Hence for every target $\delta \in [0, H(\vec p)]$ there is a unique $\rho^\star(\delta; \vec p)$ with $\Delta(\rho^\star; \vec p) = \delta$, and running the race with $\rho_t = \rho^\star(\delta_t; \vec p^{(t)})$ removes exactly $\delta_t$ at every previously unused window, keeps the token law $\vec p^{(t)}$, and leaves the detector unchanged.
\end{theorem}
Under speculative decoding, $\rho_t$ may depend on the draft distribution $\vec q^{(t)}$ but not on the target distribution $\vec p^{(t)}$: the drafter fixes its race inputs $\tilde R_v$, and hence $\rho_t$, before the target computes $\vec p^{(t)}$.

\begin{proof}
Write $A = G^\star$, $Z = \rho A + \sqrt{1-\rho^2}\,\eta$, and $\Delta(\rho) = I(A; X)$.

\emph{Step 1: endpoints.}
At $\rho = 0$, $Z = \eta$ is independent of $A$, so $X$ depends only on $\eta$ and $\Delta = 0$.
At $\rho = 1$, $Z = A$, so $X$ is a function of $A$ and $\Delta = H(X) = H(\vec p)$, since $X \sim \vec p$ (\autoref{lem:copula}).

\emph{Step 2: monotone.}
Take $0 < \rho' < \rho \leq 1$ and let $c = \rho'/\rho$.
Draw an independent $\xi \sim \mathcal N(0, I)$ and form the degraded key $A' = cA + \sqrt{1-c^2}\,\xi$, which is standard Gaussian.
Run the race at coupling $\rho$ on $A'$, giving $\tilde X$.
Since
\[
\rho A' + \sqrt{1-\rho^2}\,\eta = \rho' A + \sqrt{\rho^2 - \rho'^2}\,\xi + \sqrt{1-\rho^2}\,\eta,
\]
and the last two terms form a Gaussian of variance $(\rho^2 - \rho'^2) + (1-\rho^2) = 1 - \rho'^2$ independent of $A$, the pair $(A, \tilde X)$ has the law of $(A, X)$ at coupling $\rho'$, while $(A', \tilde X)$ has the law at coupling $\rho$.
As $A \to A' \to \tilde X$ is a Markov chain, data processing gives $\Delta(\rho') = I(A; \tilde X) \leq I(A'; \tilde X) = \Delta(\rho)$.

\emph{Step 3: strict.}
The chain rule gives $I(A';\tilde X) - I(A;\tilde X) = I(A'; \tilde X \mid A)$.
It vanishes only if the law of $\tilde X$ given $(A, A')$ depends on $A$ alone.
Given $A$, $A'$ still has a density that is positive everywhere, so this law would have to be constant in $A'$.
It is not: \rev{for a token $v$ with $p_v > 0$, }sending $A'_v \to +\infty$ drives the race input $\Phi(\rho A'_v + \sqrt{1-\rho^2}\,\eta_v)$ to 1, so $v$ wins almost surely\rev{, and since $H(\vec p) > 0$ the same holds for a second token $w$ with $p_w > 0$}. Hence $I(A'; \tilde X \mid A) > 0$ and the inequality is strict.

\emph{Step 4: continuous.}
Write $\Delta(\rho) = H(\vec p) - \E_A[H(\pi_\rho(\cdot \mid A))]$, with $H(\vec p)$ fixed.
For \revdel{each fixed}\rev{almost every} $A$, the race inputs move continuously with $\rho$, and the argmax is almost surely unique, so the conditional law $\pi_\rho(\cdot \mid A)$ moves continuously.
Entropy is continuous and bounded by $\ln |\V|$, so dominated convergence makes the expectation, hence $\Delta$, continuous.

\emph{Conclusion.}
$\rho \mapsto \Delta(\rho)$ is a continuous strictly increasing bijection from $[0,1]$ onto $[0, H(\vec p)]$.
Every $\delta$ in range is therefore hit by exactly one $\rho^\star(\delta; \vec p)$.
A prefix-dependent $\rho_t$ keeps the token law by \autoref{lem:copula}, applied given the prefix, and the scores do not involve $\rho_t$, so the detector is unchanged.
\end{proof}

\section{Near-optimality of the copula race}
\label{app:optimal}

\revdel{This appendix bounds the detection evidence that any distortion-free coupling can extract from the entropy it removes, and shows that the copula race comes close to this bound (\autoref{subsec:entropy}).
}
\paragraph{Setup.} Fix a next-token distribution $\vec p$ and a detection score $s$ applied to the emitted token's keyed Gaussian, with mean $\mu_0$ and standard deviation $\sigma_0$ under the null (unwatermarked sampling). A coupling $\pi(x \mid g)$ between the keyed Gaussians $g$ and the emitted token $x$ is assessed by its per-token signal-to-noise ratio $\mathrm{SNR} = (\E[s(g_X)] - \mu_0)/\sigma_0$ and its cost $\Delta = I(g; X)$, the entropy removed \eqref{eq:delta}.

\paragraph{An upper bound for every scheme.}
\begin{proposition}[No free evidence]
\label{prop:kl}
For any distortion-free coupling of keyed Gaussians $G$ and token $X$, let $Q$ be the law of the emitted token's Gaussian $G_X$. Then $\mathrm{KL}\big(Q \,\Vert\, \mathcal N(0,1)\big) \le I(G; X) = \Delta$, and the inequality is strict whenever $\Delta > 0$.
\end{proposition}
\revdel{The left-hand side is the expected log-likelihood ratio per token in favour of the watermark -- the rate at which detection evidence accumulates over tokens. The proposition thus says that each nat of entropy removed buys at most (in fact, strictly less than) one nat of evidence: detection efficiency is fundamentally capped by the diversity budget.}
\rev{The left-hand side is the expected log-likelihood ratio per token, i.e.\ the rate at which detection evidence accumulates, so each nat of entropy removed yields less than one nat of evidence.}

\begin{proof}
Under $\pi$, the KL to the product of its marginals is the mutual information: $\mathrm{KL}(\pi \Vert P_G \otimes \vec p) = I(G;X) = \Delta$.
The map $(g, x) \mapsto g_x$ reads off the emitted token's Gaussian.
It sends $P_G \otimes \vec p$ to $\mathcal N(0,1)$, since every coordinate of $G$ is standard normal, and $\pi$ to $Q$.
Data processing gives the bound.

For strictness, suppose equality holds with $\Delta > 0$.
Equality in data processing requires the likelihood ratio to be a function of the statistic: $\pi(x \mid g)/p_x = f_x(g_x)$ for some $f_x$.
Summing over $x$ gives $\sum_x p_x f_x(g_x) = 1$ for all $g$.
This is a sum of functions of separate independent coordinates, so each summand is constant: the variance of the sum is the sum of the variances, and it is zero.
Hence each $f_x$ is constant, and $\E_g\, \pi(x \mid g) = p_x$ makes it 1.
Then $\pi(x \mid g) = p_x$ and $\Delta = 0$.
\end{proof}

\paragraph{The best coupling.} Maximising $\E[s(g_X)]$ subject to $X \sim \vec p$ and $I(g;X) \le \Delta$ has the Lagrangian $\E_\pi[s] - \lambda\,\mathrm{KL}(\pi \Vert P_g \otimes \vec p)$, maximised by
\[
\pi(x \mid g) \propto a_x\, e^{s(g_x)/\lambda}, \qquad a \text{ set so that } X \sim \vec p,
\]
an exponential tilt, i.e.\ entropic optimal transport, which textGrain approximates per token.

\paragraph{Weak-watermark limit.} Write $\pi(x \mid g) = p_x\big(1 + \varepsilon\, \psi_x(g)\big)$ with $\E_g \psi_x = 0$ and $\sum_x p_x \psi_x = 0$.
The first constraint keeps the token marginal at $\vec p$; the second keeps $\pi(\cdot \mid g)$ normalized.
To second order, the cost is $\Delta \approx \tfrac12 \varepsilon^2 \sum_x p_x \E \psi_x^2$.
The signal is $\varepsilon \sum_x p_x \E[s(g_x)\psi_x]$.
The best perturbation is $\psi_x = \bar s(g_x) - \sum_v p_v \bar s(g_v)$, with $\bar s = s - \mu_0$, and it gives, for every score,
\[
\mathrm{SNR} \le \sqrt{2\Delta\,\big(1 - \textstyle\sum_v p_v^2\big)}.
\]
\paragraph{Copula efficiency.} For the copula, $\Delta \approx \tfrac12 J(\vec p)\rho^2$ with $J(\vec p) = \sum_v p_v\, e(1/p_v)^2/(1-p_v)$, where $e(n)$ is the mean maximum of $n$ standard normals\rev{, extended to real $n$ as the mean of the law with CDF $\Phi^n$}.
\revdel{The signal}\rev{Its $\mathrm{SNR}$} is $\rho\,\mathrm{Cov}(s(g), g) \sum_v p_v e(1/p_v)/\sigma_0$.
Its efficiency, the ratio of its SNR to the bound, is
\[
\eta(\vec p) = \mathrm{corr}\big(s(g), g\big) \cdot \frac{\sum_v p_v\, e(1/p_v)}{\sqrt{J(\vec p)\,(1 - \sum_v p_v^2)}} \le 1
\]
by Cauchy--Schwarz, with equality iff $s$ is \revdel{linear}\rev{increasing and affine} in $g$ and $\vec p$ is uniform on its support.
Extending $e(n)$ to non-integer $n$ is accurate: an exact computation from the score identity $\partial \Prob(X = x)/\partial \mu_v = \E[\mathbb 1\{X = x\}\, \eta_v]$ agrees with $\eta(\vec p)$ to within 0.001 on uniform, Zipf, and peaked distributions.

\revdel{\paragraph{Numbers.} For 8 equiprobable tokens, the copula achieves $\mathrm{KL}/\Delta = 0.85$, $0.79$ and $0.71$ at $\rho = 0.5$, $0.8$ and $0.95$ respectively, against a weak-limit ceiling factor of $1 - \sum_v p_v^2 = 0.88$. Over 5{,}169 positions of 60 Qwen3.5-9B answers with at least two candidates (after top-$p$), the weak-limit efficiency $\eta$ with the Gaussian score averages 0.966 -- 0.971 when weighted by each position's capacity $1 - \sum_v p_v^2$ -- with a 10th percentile of 0.94. Rescoring generated Qwen3.5-9B answers under the copula and \tsii with the Gaussian score instead of \ts's score changes detection by at most 2 percentage points, so the choice of score is not the bottleneck. At finite strength, \autoref{tab:optimal} compares the copula with the optimal coupling at equal $\Delta$ (computed by Sinkhorn on 4000 sampled keys, a discretisation that favours the optimum): the copula attains 91--100\% of the optimal signal-to-noise ratio for $\rho \geq 0.5$.}
\rev{\paragraph{Numbers.} For 8 equiprobable tokens, the copula achieves $\mathrm{KL}/\Delta = 0.85$, $0.79$ and $0.71$ at $\rho = 0.5$, $0.8$ and $0.95$ respectively, against a weak-limit ceiling factor of $1 - \sum_v p_v^2 = 0.88$. Over 5{,}169 positions of 60 Qwen3.5-9B answers with at least two candidates (after top-$p$), the weak-limit efficiency $\eta$ with the Gaussian score averages 0.966 (0.971 when weighted by each position's capacity $1 - \sum_v p_v^2$), with a 10th percentile of 0.94. Rescoring generated Qwen3.5-9B answers under the copula and \tsii with the Gaussian score instead of \ts's score changes detection by at most 2 percentage points. At finite strength, \autoref{tab:optimal} compares the copula with the optimal coupling at equal $\Delta$ (computed by Sinkhorn on 4000 sampled keys, a discretisation that favours the optimum): the copula attains 91--100\% of the optimal signal-to-noise ratio at $\rho = 0.5$, $0.8$ and $0.95$.}

\begin{table}[ht]
\centering
\footnotesize
\caption{Copula race vs the best distortion-free coupling at equal entropy removed: ratio of per-token signal-to-noise ratios, for the Gaussian score $s = g$ and \ts's score $s = -\ln(1 - u)$, at couplings $\rho = 0.5 / 0.8 / 0.95$ (the fraction of each token's entropy removed by the key is then about 0.16 / 0.45 / 0.72).}
\label{tab:optimal}
\begin{tabular}{lcc}
\toprule
Next-token distribution & Gaussian score & \ts score \\
\midrule
uniform over 8 tokens & 0.998 / 0.998 / 1.000 & 0.915 / 0.962 / 0.992 \\
Zipf(1.1) over 24 tokens & 0.982 / 0.988 / 0.993 & 0.908 / 0.960 / 0.989 \\
top token 0.5, Zipf tail & 0.966 / 0.978 / 0.987 & 0.911 / 0.959 / 0.990 \\
top token 0.8, Zipf tail & 0.929 / 0.938 / 0.942 & 0.908 / 0.952 / 0.973 \\
\bottomrule
\end{tabular}
\end{table}

\section{Entropy removed by a constant coupling}
\label{app:entropy}

\revdel{This appendix measures the fraction of each token's entropy that a fixed $\rho$ removes, on synthetic next-token distributions.

\autoref{tab:entropy_frac} reports the fraction of entropy removed, $\Delta(\rho; \vec p)/H(\vec p)$, computed by quadrature over the keyed Gaussians (640 sampled keys per value), for synthetic next-token distributions: uniform over $n$ tokens, a flat Zipf(0.8) law over 100 tokens, and a top token of probability $p_{\max}$ with a Zipf(1.1) tail over 99 further tokens. The fraction depends only weakly on the shape: a fixed $\rho$ removes 0.47--0.51 of the entropy at $\rho = 0.8$ and 0.63--0.65 at $\rho = 0.9$ on the many-token shapes, and about 0.1 less in the extreme case of two equiprobable tokens. A single $\rho$ therefore acts as an approximately shape-independent budget fraction, which is what \autoref{subsec:entropy} relies on.}
\rev{\autoref{tab:entropy_frac} reports the fraction of entropy removed, $\Delta(\rho; \vec p)/H(\vec p)$, computed by quadrature over the keyed Gaussians (640 sampled keys per value; exact for two tokens), for synthetic next-token distributions: uniform over $n$ tokens, a flat Zipf(0.8) law over 100 tokens, and a top token of probability $p_{\max}$ with a Zipf(1.1) tail over 99 further tokens. The fraction depends only weakly on the shape: a fixed $\rho$ removes 0.47--0.49 of the entropy at $\rho = 0.8$ and 0.63--0.65 at $\rho = 0.9$ on the many-token shapes, and about 0.1 less in the extreme case of two equiprobable tokens. A single $\rho$ therefore removes a nearly shape-independent fraction of the entropy, as \autoref{subsec:entropy} uses.}

\begin{table}[t!]
\centering
\footnotesize
\caption{Fraction of each token's entropy removed at a fixed coupling $\rho$, on synthetic distributions.}
\label{tab:entropy_frac}
\begin{tabular}{lcccccc}
\toprule
Shape & $H$ (nats) & $\rho = 0.5$ & 0.7 & 0.8 & 0.9 & 0.95 \\
\midrule
uniform over 2 & 0.69 & \revdel{0.122}\rev{0.125} & \revdel{0.266}\rev{0.272} & \revdel{0.380}\rev{0.385} & 0.550 & \revdel{0.686}\rev{0.677} \\
uniform over 8 & 2.08 & 0.147 & 0.308 & 0.426 & 0.590 & 0.712 \\
uniform over 64 & 4.16 & 0.173 & 0.353 & 0.477 & 0.640 & 0.754 \\
flat Zipf(0.8) & 4.07 & 0.173 & 0.354 & 0.479 & 0.641 & 0.750 \\
$p_{\max} = 0.5$ & 2.42 & 0.170 & 0.350 & 0.475 & 0.638 & 0.750 \\
$p_{\max} = 0.95$ & 0.37 & 0.183 & 0.370 & 0.493 & 0.652 & 0.757 \\
\bottomrule
\end{tabular}
\end{table}

These fractions are characterised on synthetic shapes; on real Qwen3.5-9B distributions we instead measure the copula's efficiency directly (\autoref{app:optimal}), which is the quantity that enters detection.

\section{Algorithm and latency}
\label{app:latency}
\revdel{
This appendix gives one round of \tsii in pseudocode and measures the post-logits latency of each sampler.}

\subsection{One round in pseudocode}
\label{app:latency-alg}

\autoref{alg:sampler} implements one round of speculative decoding with the residual race (\autoref{subsec:residual}); its last line gives decoding without speculation.

\begin{algorithm}[ht]
\caption{One round of \tsii with speculative decoding ($N$ drafts)}
\label{alg:sampler}
\begin{algorithmic}[1]
\Require prefix $x$, keys $k^{(1)}, k^{(2)}$, coupling $\rho$, drafter and target, number of drafts $N$, set $\mathcal W$ of windows raced so far
\For{$i = 1, \dots, N+1$}
\State $\mathbf w_i \gets$ last $k$ tokens of $(x, d_{<i})$
\If{$\mathbf w_i \in \mathcal W$} \State $\tilde R_v \gets$ freshly drawn $\Unif$ values \Comment{window already raced}
\Else \State $G_v^\star \gets \big(\Phi^{-1}(R_v^{(1)}) + \Phi^{-1}(R_v^{(2)})\big)/\sqrt 2$ at $\mathbf w_i$; \quad $\tilde R_v \gets \Phi\big(\rho G_v^\star + \sqrt{1-\rho^2}\,\eta_v\big)$ \Comment{key-only}
\EndIf
\State $E_{i,v} \gets -\ln \tilde R_v$; \quad $\mathcal W \gets \mathcal W \cup \lbrace \mathbf w_i \rbrace$
\If{$i \le N$}
\State $\vec q_i \gets$ drafter's law after $(x, d_{<i})$
\State $d_i \gets \argmin_v E_{i,v}/q_{i,v}$; \quad $\tau_i \gets E_{i,d_i}/q_{i,d_i}$ \Comment{draft race \eqref{eq:draft}}
\EndIf
\EndFor
\State $\vec p_1, \dots, \vec p_{N+1} \gets$ target's laws after $(x, d_{<i})$, in one forward pass
\State $\pi_{i,v} \gets \Prob(I = i)\, r_{i,v}$ from \eqref{eq:stoplaw}; \quad $\pi_{i,d_i} \gets 0$ for $i \le N$ \Comment{block verification}
\State $(I, y) \gets \argmin_{(i,v):\, \pi_{i,v} > 0} (E_{i,v} - \tau_i q_{i,v})/\pi_{i,v}$, with $\tau_{N+1} q_{N+1,v} = 0$ \Comment{residual race \eqref{eq:relay}}
\State \Return $d_{<I}\, y$
\Statex \textit{Without speculative decoding: return $\argmin_v E_v/p_v$ at the window of the next position.}
\end{algorithmic}
\end{algorithm}

In \autoref{alg:sampler}, the masses $\pi_{i,d_i}$ vanish in exact arithmetic (\autoref{lem:zeromass}) but are set to 0 explicitly: the leftover of a draft's own lane is 0, so a rounding residue in its mass would make it win the race.
The key-only work of position $i$ needs the draft $d_{i-1}$, so it runs alongside the drafter's forward \revdel{pass at that position;} \rev{pass;} after the \revdel{target's forward} \rev{target} pass, block verification's law and the relay \revdel{take work} \rev{are} linear in the \rev{round's} $(N+1)\,|\V|$ \revdel{values of the round.} \rev{values.}

\subsection{Post-logits latency}
\label{app:latency-post}

\revdel{This subsection times the post-logits work of each sampler, and the full decoding step of \tsii against that of \ts.}
Only work performed after the logits are known can add latency; key-only work (PRF evaluation, forming $G^\star$ and $\eta$, and textGrain's keyed cost table) depends only on the key and the window, and can therefore run concurrently with the model's forward pass.
\autoref{tab:latency} accordingly times the post-logits work of each sampler on the Qwen3.5 vocabulary, with every method except our full textGrain implementation captured in a CUDA graph.
\tsii costs the same as \ts: a single pass over the vocabulary, cheaper than the top-$p$ filtering that every method (watermarked or not) requires.
For textGrain we \revdel{report} \rev{time} three \revdel{variants of increasing fidelity:} \rev{variants:} a minimal version \revdel{implementing} \rev{with} only the core post-logits work (block masses, a fixed number of Sinkhorn iterations, keyed column selection, and two sampling \revdel{draws)} \rev{draws),} without calibrating the regularisation to the entropy budget; \revdel{the} \rev{our} \revdel{published} \rev{implementation of its} Algorithm A.1 with 10--30 Sinkhorn iterations; and \revdel{our} \rev{the} \revdel{full} \rev{same} \revdel{implementation,} \rev{implementation} \revdel{which additionally calibrates} \rev{with} the regularisation \rev{calibrated} per token so that the budget is met. Any faithful implementation \revdel{must do} \rev{does} at least the minimal version's work.
At batch size 1, \rev{our implementation of} Algorithm A.1 takes 1.1--3.0\,ms per step \revdel{-- 2--6\%} \rev{(2--6\%} of a 50\,ms Qwen3.5-9B decoding \revdel{step} \rev{step,} \revdel{(\autoref{tab:ts12}) --} \rev{\autoref{tab:ts12}),} and \revdel{our} \rev{2.3--4.1\,ms} \revdel{full} \rev{with} \revdel{implementation 2.3--4.1\,ms.} \rev{calibration.} Even the minimal version (0.2--0.4\,ms at batch 1) \revdel{requires} \rev{does} 8--22 times \tsii's post-logits work \revdel{across batch sizes 1, 8 and 64} (\autoref{tab:latency}).

\begin{table}[t!]
\centering
\footnotesize
\caption{Post-logits latency per decoding step (ms, median of 200 CUDA-graph replays, or 50 eager runs for \revdel{our} \rev{calibrated} \revdel{textGrain implementation,} \rev{textGrain,} whose calibration loop has data-dependent control flow; one H200, vocabulary of 248{,}320 tokens, synthetic high-entropy logits). Top-$p$ filtering is shared by all methods. The minimal textGrain version uses 5--20 Sinkhorn iterations on the $256 \times 256$ kernel without calibrating $\lambda$.}
\label{tab:latency}
\begin{tabular}{lcccccc}
\toprule
& & & & \multicolumn{3}{c}{textGrain} \\
\cmidrule(l){5-7}
Batch & Top-$p$ (all) & \ts & \tsii & minimal & A.1, 10 / 20 / 30 it. & \revdel{ours,} \rev{A.1 calibrated,} $\beta = 0.25 / 0.5$ \\
\midrule
1 & 0.16 & 0.020 & 0.018 & 0.22--0.40 & 1.1 / 2.1 / 3.0 & 2.3 / 4.1 \\
8 & 0.72 & 0.030 & 0.025 & 0.28--0.50 & 1.4 / 2.6 / 3.8 & 2.4 / 4.5 \\
64 & 2.14 & 0.15 & 0.12 & 0.98--1.24 & 3.8 / 6.7 / 9.5 & 2.7 / 4.8 \\
\bottomrule
\end{tabular}
\end{table}

\paragraph{Comparison with \ts.}
\tsii needs both keys' PRF values and their Gaussian transforms at every token, whereas \ts needs a single key's PRF values.
\autoref{tab:ts12} times full Qwen3.5-9B decoding steps (Hugging Face forward pass, one H200, real logits). If the extra key-only work is run sequentially after the forward pass, it costs 0.2\,ms at batch size 1 and 2.9\,ms (5\%) at batch size 64; overlapped on a side CUDA stream during the forward pass, \tsii is within 0.3\,ms (0.6\%) of \ts and within 0.9\% of plain unwatermarked sampling.

\begin{table}[ht]
\centering
\footnotesize
\caption{Full decoding step (ms, median), Qwen3.5-9B; overlapped: key-only work on a side CUDA stream.}
\label{tab:ts12}
\begin{tabular}{lccccc}
\toprule
Batch & Plain sampling & \ts & \tsii & \ts, overlapped & \tsii, overlapped \\
\midrule
1 & 50.44 & 50.85 & 51.08 & 50.69 & 50.87 \\
8 & 52.39 & 52.54 & 52.64 & 52.52 & 52.79 \\
64 & 56.46 & 58.43 & 61.34 & 56.81 & 56.93 \\
\bottomrule
\end{tabular}
\end{table}

\section{PRF resolution}
\label{app:prf}

\revdel{This appendix shows that the 13-bit PRF of \ts slightly distorts the race, and that the 53-bit PRF used by \tsii removes the distortion.
}
\ts's PRF returns $(h \bmod 8191)/8191$, i.e.\ values quantised to 13 bits.
The largest attainable value is then $8190/8191$, so $E_v = -\ln R_v$ is bounded away from zero. A rare token can win the race only if its $E_v$ is extremely small relative to $p_v$, and is therefore emitted slightly too rarely: the quantisation measurably distorts the sampler.
Simulating the race at all 8059 positions of 60 Qwen3.5-9B answers, the mean of $-\ln p_x - H(\vec p)$ is $-0.03 \pm 0.12$ millinats per token with continuous (unquantised) uniform values, against $-1.00 \pm 0.14$ millinats with the 13-bit PRF. The bias concentrates at high-entropy positions: it is $-0.3$ millinats at positions with 2--10 candidates after top-$p$ (52\% of positions), but $-42$ millinats at positions with 101--1000 candidates (0.9\%).
\tsii therefore uses a \revdel{PRF} \rev{53-bit} \revdel{with substantially more output bits.} \rev{PRF.}
With the 53-bit PRF (splitmix64; uniformity checked by a Kolmogorov--Smirnov test), \tsii under speculative decoding passes both distortion tests of \autoref{fig:ablation}.
The 13-bit PRF also induces a slight correlation between the two keys' values ($-0.011 \pm 0.001$ on random windows and tokens, and $-0.019 \pm 0.002$ on unwatermarked answers, against $-0.0006$ for splitmix64), so that $G^\star$ has variance 0.98--0.99 rather than 1. With 13 bits, \tsii's mean per-token score was 1.2--1.6\% below the single-key copula's on the two Qwen models at $\rho = 0.9$; with the 53-bit PRF the two agree on all three models at $\rho = 0.8$ and $0.9$ (differences of per-answer mean scores from $-0.006$ to $+0.011$, every 95\% interval containing zero), as they must, since the two samplers coincide in distribution.

\section{Additional results}
\label{app:more}

\revdel{This appendix collects the results that \autoref{sec:experiments} refers to: more keys, a lower temperature, entropy weighting, more models and lengths, robustness, keys with different windows, repeated windows, the number of drafts per round, a second drafter--target pair, math problems under speculative decoding, and code.

\paragraph{Four keys.} With \ts's (13-bit) PRF, merging four keys detects 63\% / 91\% of 100 / 200-token answers, against 60\% / 92\% for two keys: additional keys -- e.g.\ one publishable key per verifier -- cost the full detector no detection power, as \autoref{subsec:copula} predicts.}
\rev{\paragraph{Four keys.} With \ts's (13-bit) PRF, merging four keys detects 63\% / 91\% of 100 / 200-token answers, against 60\% / 92\% for two keys: additional keys (e.g.\ one publishable key per verifier) cost the full detector no detection power, as \autoref{subsec:copula} predicts.}

\paragraph{Temperature 0.7.} At the lower entropies induced by temperature 0.7, \tsii detects 22\% / 56\% of 100 / 200-token answers, against 9\% / 37\% for \ts at a similar Self-BLEU (26.3 vs.\ 25.4); textGrain detects 22\% / 55\% at a Self-BLEU of 27.3. (For this experiment, Self-BLEU uses 16 answers to each of 25 questions, and thresholds are calibrated on 6000 unwatermarked answers.)

\begin{table}[ht]
\centering
\footnotesize
\caption{Key-only vs.\ entropy-weighted detection (\ts's weighted Gamma test): \% of 100 / 200-token answers detected at a false-positive rate of $10^{-3}$, with thresholds calibrated per test over 10 keys and 10{,}000 unwatermarked answers. First 1000 answers per configuration without speculative decoding; last block: all 3000 answers with speculative decoding (8 drafts, block verification). Strengths: $\alpha = 0.5$ (\ts), $\beta = 0.5$ (textGrain), $\rho = 0.9$ (copulas, 53-bit PRF). For SynthID-Text, the key-only test is its official binomial test and the weighted test is the same Gamma test applied to twice its mean coin score. Self-BLEU over 32 answers to each of 100 questions. Weighting improves every method by 4--23 percentage points at 200 tokens, and the two copulas remain level without speculative decoding, as they must in distribution.}
\label{tab:entropy}
\begin{tabular}{llccc}
\toprule
Model & Method & Self-BLEU & Key-only & Entropy-weighted \\
\midrule
\multirow{4}{*}{Llama-3.1-8B} & \ts & 23.9 & 53 / 82 & 58 / 91 \\
 & textGrain & 24.1 & 55 / 86 & 67 / 94 \\
 & AutoMark copula & 23.9 & 63 / 90 & 71 / 95 \\
 & \tsii & 23.9 & 64 / 90 & 71 / 95 \\
\midrule
\multirow{5}{*}{Qwen3.5-9B} & \ts & 20.6 & 42 / 82 & 59 / 93 \\
 & textGrain & 20.9 & 65 / 91 & 74 / 97 \\
 & SynthID-Text (10 layers) & 20.9 & 54 / 92 & 72 / 98 \\
 & AutoMark copula & 20.9 & 64 / 93 & 77 / 97 \\
 & \tsii & 20.7 & 63 / 91 & 78 / 97 \\
\midrule
\multirow{4}{*}{Qwen3.5-27B} & \ts & 26.5 & 21 / 53 & 32 / 76 \\
 & textGrain & 27.2 & 37 / 71 & 44 / 84 \\
 & AutoMark copula & 28.2 & 37 / 73 & 48 / 87 \\
 & \tsii & 28.5 & 36 / 71 & 47 / 87 \\
\midrule
\multirow{6}{*}{\shortstack[l]{Qwen3.5-9B,\\spec.\ dec.}} & \ts & 20.1 & 31 / 71 & 48 / 87 \\
 & textGrain, second key & 18.9 & 24 / 59 & 36 / 76 \\
 & textGrain, own recipe$^\ast$ & 20.6 & 59 / 89 & 69 / 96 \\
 & SynthID-Text (10 layers) & 18.4 & 14 / 41 & 23 / 63 \\
 & AutoMark copula, second key & 19.0 & 23 / 55 & 36 / 78 \\
 & \tsii, residual race & 19.1 & 55 / 87 & 67 / 95 \\
\bottomrule
\end{tabular}
\\[2pt]\parbox{\linewidth}{\raggedright\scriptsize $^\ast$Emits 5\% fewer tokens per round; under speculative decoding, Self-BLEU over 16 answers per question.}
\end{table}

\begin{figure}[b!]
\centering
\includegraphics[width=\linewidth]{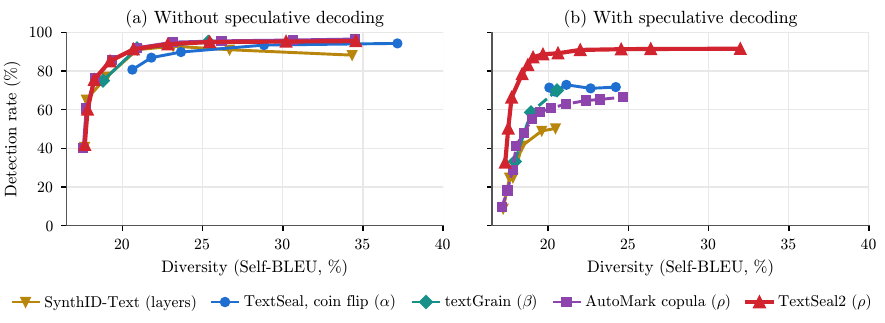}
\caption{The diversity--detection trade-off of \autoref{fig:teaser} at 200 tokens instead of 100; same protocol and conventions, including the speculative panel (b).}
\label{fig:pareto200}
\end{figure}

\begin{figure}[ht]
\centering
\includegraphics[width=\linewidth]{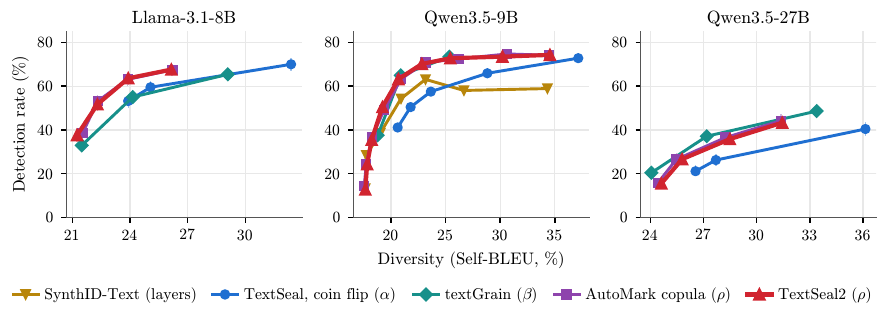}
\caption{Diversity (Self-BLEU; lower is more diverse) vs.\ detection of 100-token answers without speculative decoding, on three models (32 answers to each of 100 questions; copulas with the 53-bit PRF). Interpolated at \ts's Self-BLEU, \tsii gains 10, 22 and 8 percentage points over the coin flip on Llama-3.1-8B, Qwen3.5-9B and Qwen3.5-27B respectively, and AutoMark's copula overlaps \tsii throughout.}
\label{fig:three_models}
\end{figure}

\begin{table}[ht]
\centering
\footnotesize
\caption{Robustness: \% of 1000 200-token Qwen3.5-9B answers still detected after each attack, at the clean thresholds ($10^{-3}$ false-positive rate, 10 keys; paraphraser: \mbox{Llama-3.2-3B-Instruct}).}
\label{tab:robust}
\begin{tabular}{lccccc}
\toprule
Method & Clean & 10\% subst. & 20\% subst. & 30\% var. & Paraphr. \\
\midrule
TextSeal (coin flip, $\alpha = 0.5$) & 82 & 46 & 16 & 28 & 3 \\
textGrain ($\beta = 0.5$) & 91 & 63 & 24 & 48 & 5 \\
AutoMark's copula ($\rho = 0.9$) & 92 & 66 & 26 & 48 & 5 \\
\tsii ($\rho = 0.9$) & 92 & 64 & 27 & 46 & 6 \\
\tsii ($\rho = 0.95$) & 94 & 71 & 31 & 51 & 7 \\
\bottomrule
\end{tabular}
\end{table}

\paragraph{Keys with different windows.}
The virtual key can merge keys seeded by windows of different lengths. Let $F_t$ be the set of keys whose window at position $t$ has not been used before in the answer; then $G^\star_v = \sum_{j \in F_t} G_{j,v}/\sqrt{|F_t|}$ is still standard normal, so the race remains distortion-free with any such mix.
This enables edit-rate estimation without access to the original text\revdel{~\citep{milis2026public}}: under token substitutions at rate $e$, a window of length $L$ \revdel{survives} \rev{and the token it seeds survive} intact with probability \revdel{approximately} \rev{about} \revdel{$(1-e)^L$,} \rev{$(1-e)^{L+1}$~\citep{milis2026public},} so the mean excess score of the length-$L$ key decays geometrically in $L$, at a rate that identifies \revdel{$e$.} \rev{$e$; the least-squares slope of the log excess against $L+1$ estimates $\ln(1-e)$.}
\revdel{Merging} \rev{This} \revdel{is} \rev{application} \revdel{essential} \rev{needs} \revdel{for this application.} \rev{merging.} Under \ts's coin-flip routing, each of $M$ keys generates only a $1/M$ share of the tokens, and detection collapses as keys are added: a ladder of window lengths $1, \dots, 10$ detects just 22\% of 400-token Llama-3.1-8B-Instruct answers. Merged into $G^\star$, each key instead retains a \revdel{$1/\sqrt M$}\rev{$\rho/\sqrt M$} correlation with every token's race, and the same ladder detects 98\%; on Qwen3.5-9B, multi-window variants match two-key \tsii at equal Self-BLEU on clean text (\autoref{tab:multiwin}).
Pooled over a corpus, the resulting estimate of $e$ is unbiased up to 20\% token substitutions; for a single 200-token answer it remains informative up to about 5\%.
The cost is robustness: long windows are broken by edits first, and under 8--12\% substitutions the ladder's detection falls by 15--27 percentage points relative to fixed short windows.

\begin{table}[ht]
\centering
\footnotesize
\caption{Keys with different window lengths merged in the virtual key (Qwen3.5-9B, ELI5, 200 tokens; $\rho = 0.9$ for the robustness and estimation columns; 53-bit PRF). Rows: standard two-key \tsii with windows $(3,3)$; two windows $(2,10)$; and a six-key ladder $(2,3,4,6,8,10)$. Clean: \% of 100 / 200-token answers detected at a false-positive rate of $10^{-3}$, interpolated at Self-BLEU 21.5. Substitution columns: \% of 200-token answers still detected, at the clean thresholds, after a fraction $e$ of tokens is replaced uniformly at random. Estimate columns: edit rate fitted from the decay of per-key scores with window length, reported as (pooled fit over 1000 answers) / (median of per-answer fits).}
\label{tab:multiwin}
\setlength{\tabcolsep}{4pt}
\begin{tabular}{lcccccc}
\toprule
Windows & Clean & $e = 0.05$ & $e = 0.12$ & $e = 0.2$ & Estimate at $e = 0.05$ & Estimate at $e = 0.12$ \\
\midrule
$(3,3)$ & 68 / 92 & 83 & 58 & 27 & -- & -- \\
$(2,10)$ & 70 / 94 & 74 & 31 & 10 & 0.045 / 0.043 & 0.116 / 0.087 \\
ladder & 67 / 94 & 75 & 33 & 9 & 0.045 / 0.037 & 0.117 / 0.059 \\
\bottomrule
\end{tabular}
\end{table}

\paragraph{Repeated windows in long answers.}
In long-form prose, repeated windows are rare: on 1024-token Qwen3.5-9B answers (prompts requesting at least 800 words; 300 answers per scheme; thresholds calibrated on 2000 unwatermarked long answers), only 2.3\% of positions are a window's second occurrence. Every scheme \revdel{-- \ts,} \rev{(\ts,} \tsii with one use per window, and AutoMark's \revdel{copula --} \rev{copula)} then detects 98--100\% of answers within 256 tokens and 100\% from 512 tokens (at a 1\% false-positive rate), so the \revdel{generation} \rev{repeated-window} rule \revdel{for} \rev{does} \revdel{repeated} \rev{not} \revdel{windows is immaterial} \rev{matter} there.
Windows repeat far more often in code (11\% of positions, \autoref{tab:code}) and in thinking traces (16--34\%, \autoref{tab:repeat}), \revdel{which is} where leftover continuation helps.

\begin{table}[ht]
\centering
\footnotesize
\caption{Number of drafts per round: Qwen3.5-0.8B drafting $N$ tokens for Qwen3.5-9B, block verification, every configuration generated separately for each $N$ (3000 answers for detection, 100 questions $\times$ 16 answers for Self-BLEU). Tokens per round: emitted tokens over target passes, at $\alpha = 0.5$, $\beta = 0.5$, 10 layers and $\rho = 0.9$. Detected: \% of 100 / 200-token answers at a false-positive rate of $10^{-3}$, interpolated at \ts's Self-BLEU for the same $N$ (19.8 / 20.1 / 19.6). Paired over 480 prompts, only textGrain's own recipe lowers the tokens per round ($-0.08$ / $-0.14$ / $-0.18$); the distortion tests of \autoref{fig:ablation} detect the raw-lane-length control at every $N$, and no other scheme on the tests that apply to it, except that at $N = 16$ the residual race and textGrain's own recipe both sit 7 millinats below the same unwatermarked run in the all-token test (intervals $[-13, -1]$), which points to noise in that shared reference.}
\label{tab:ndrafts}
\setlength{\tabcolsep}{4pt}
\begin{tabular}{lcccccc}
\toprule
& \multicolumn{2}{c}{$N = 4$} & \multicolumn{2}{c}{$N = 8$} & \multicolumn{2}{c}{$N = 16$} \\
\cmidrule(lr){2-3} \cmidrule(lr){4-5} \cmidrule(lr){6-7}
Method & Tokens / round & Detected & Tokens / round & Detected & Tokens / round & Detected \\
\midrule
No watermark & 2.44 & -- & 2.74 & -- & 2.85 & -- \\
\ts (coin flip) & 2.46 & 34 / 73 & 2.77 & 31 / 71 & 2.87 & 30 / 68 \\
textGrain & 2.45 & 30 / 65 & 2.75 & 31 / 67 & 2.84 & 28 / 62 \\
textGrain, own recipe & 2.36 & 48 / 82 & 2.61 & 52 / 84 & 2.68 & 51 / 85 \\
SynthID-Text & 2.44 & 19 / 48 & 2.75 & 17 / 50 & 2.85 & 17 / 46 \\
AutoMark copula & 2.44 & 24 / 57 & 2.73 & 25 / 60 & 2.83 & 23 / 57 \\
\tsii, residual race & 2.44 & \textbf{58} / \textbf{88} & 2.74 & \textbf{59} / \textbf{89} & 2.85 & \textbf{53} / \textbf{86} \\
\bottomrule
\end{tabular}
\end{table}

\begin{table}[ht]
\centering
\footnotesize
\caption{Second drafter--target pair: Llama-3.2-1B-Instruct drafting $N = 4$ or $8$ tokens per round for Llama-3.1-8B-Instruct, with block verification (ELI5, 240 prompts per row, $\alpha = 0.5$, $\rho = 0.9$, $\beta = 0.5$, \ts's 13-bit PRF for every scheme). Tokens per round: emitted tokens over target passes. Detection: \% of 100 / 200-token answers at a false-positive rate of $10^{-3}$, thresholds over 10 keys and 10{,}000 unwatermarked Llama-3.1-8B-Instruct answers; 95\% Wilson intervals are about $\pm 6$ points.}
\label{tab:llamapair}
\resizebox{\linewidth}{!}{\begin{tabular}{lccc}
\toprule
Method & Tokens per round ($N = 4$ / $8$) & Detected, $N = 4$ & Detected, $N = 8$ \\
\midrule
No watermark & 3.20 / 4.08 & -- & -- \\
\ts (coin flip, other key for the last token) & 3.21 / 4.09 & 36 / 72 & 37 / 62 \\
textGrain, other key for the last token & 3.21 / 4.14 & 20 / 51 & 18 / 43 \\
textGrain, own recipe (lowers acceptance) & 3.11 / 3.88 & 55 / 83 & 50 / 79 \\
\tsii, residual race & 3.20 / 4.11 & \textbf{63} / \textbf{87} & \textbf{55} / \textbf{81} \\
\bottomrule
\end{tabular}}
\end{table}

\begin{table}[ht]
\centering
\footnotesize
\caption{AIME under speculative decoding: Qwen3.5-0.8B drafting 8 tokens per round for Qwen3.5-9B with block verification, 16 answers to each of 100 AIME problems (thinking disabled, prompt asking for a \texttt{\textbackslash boxed\{\}} answer, temperature 1, top-$p$ 0.95, up to 8192 tokens), at $\alpha = 0.5$, $\beta = 0.5$, 10 layers and $\rho = 0.9$. Tokens per round: emitted tokens over target passes; only textGrain's own recipe lowers it ($-0.11$ $[-0.14, -0.09]$, paired by problem with unwatermarked decoding). Accuracy: last \texttt{\textbackslash boxed\{\}} answer; paired by problem, every scheme is within 2 points of unwatermarked decoding, all 95\% intervals containing zero. Self-BLEU: over each problem's 16 answers, first 1024 tokens. Finished: answers that end before 8192 tokens. Detected: \% of answers detected at a false-positive rate of $10^{-3}$ on their first 256 / 512 / 1024 / 2048 tokens, with each method's detector; thresholds are the empirical $10^{-3}$ quantile over the 1600 unwatermarked answers scored under 10 keys ($1.6 \cdot 10^4$ null p-values; realised false-positive rates 0.05--0.15\%); 95\% Wilson intervals are at most $\pm 2.4$ points. Strengths are defaults, not diversity-matched; Self-BLEU \revdel{stays} \rev{varies} \revdel{within} \rev{by at most} 2 \revdel{points across schemes.} \rev{points.}}
\label{tab:aime}
\setlength{\tabcolsep}{4pt}
\begin{tabular}{lcccccccc}
\toprule
& & & & & \multicolumn{4}{c}{Detected (\%) at} \\
\cmidrule(l){6-9}
Method & Tokens / round & Accuracy (\%) & Self-BLEU & Finished (\%) & 256 & 512 & 1024 & 2048 \\
\midrule
No watermark & 5.42 & 71.1 & 44.8 & 68 & -- & -- & -- & -- \\
\ts (coin flip) & 5.43 & 69.2 & 46.1 & 66 & 1 & 5 & 13 & 22 \\
textGrain & 5.39 & 69.3 & 45.8 & 66 & 1 & 2 & 6 & 15 \\
textGrain, own recipe & 5.30 & 69.7 & 46.8 & 68 & 3 & 8 & 22 & 38 \\
SynthID-Text & 5.39 & 69.2 & 45.6 & 68 & 2 & 4 & 9 & 19 \\
AutoMark copula & 5.38 & 70.2 & 45.8 & 68 & 1 & 2 & 4 & 7 \\
\tsii, residual race & 5.39 & 70.6 & 46.5 & 68 & \textbf{4} & \textbf{11} & \textbf{24} & \textbf{41} \\
\bottomrule
\end{tabular}
\end{table}

\begin{figure}[ht]
\centering
\includegraphics[width=\linewidth]{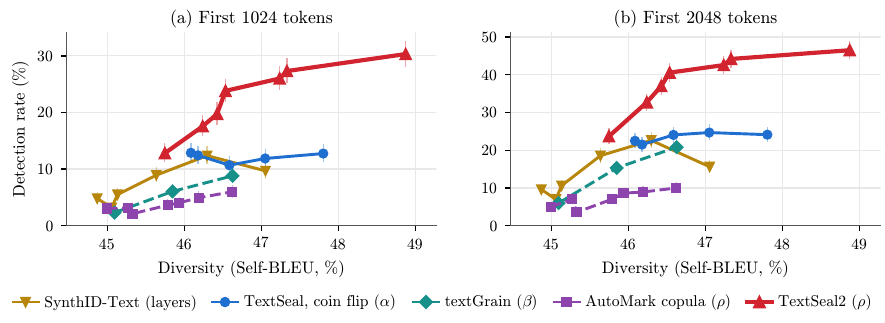}
\caption{Diversity (Self-BLEU over each problem's 16 answers, first 1024 tokens; lower is more diverse) vs.\ detection at a false-positive rate of $10^{-3}$ on AIME under speculative decoding, on the first 1024 (a) and 2048 (b) tokens; same setup as \autoref{tab:aime}, sweeping each method's strength. Interpolated at \ts's Self-BLEU (46.1), \tsii detects 30\% of answers at 2048 tokens, against 23\% for the coin flip, 21\% for SynthID-Text, 17\% for textGrain and 9\% for AutoMark's copula; paired by problem, its lead is significant over all but the coin flip ($+7.4$ $[-0.2, +19.5]$). Error bars: 95\% Wilson intervals.}
\label{fig:aime}
\end{figure}

\begin{table}[ht]
\centering
\footnotesize
\caption{Code: Qwen3.5-9B writing Python solutions to the 500 MBPP test problems, up to 512 new tokens (most answers end earlier, so detection saturates after 256 tokens). \% of answers detected at 128 / 256 / 512 tokens, thresholds at a 1\% false-positive rate over 10 keys and 2922 unwatermarked answers (all 974 MBPP problems, 3 seeds); same settings as on ELI5 ($\alpha = 0.5$, $\rho = 0.9$, $\beta = 0.5$, 10 SynthID layers; the copulas use the 53-bit PRF), whose Qwen3.5-9B Self-BLEU is 20.6--20.9 for all five; one answer per problem, so code diversity is not measured. 11\% of positions are a window's second occurrence in the text.}
\label{tab:code}
\begin{tabular}{lc}
\toprule
Method & Detected (128 / 256 / 512 tokens) \\
\midrule
\ts (coin flip) & 30 / 39 / 42 \\
textGrain & 38 / 46 / 48 \\
SynthID-Text & 42 / 50 / 50 \\
AutoMark copula & \textbf{54} / \textbf{59} / 59 \\
\tsii, one use per window & 47 / 55 / \textbf{60} \\
\bottomrule
\end{tabular}
\end{table}

\begin{table}[ht]
\centering
\footnotesize
\caption{Repeated windows in long generations: Qwen3.5-9B with thinking traces, up to 2048 tokens (ELI5: 1024 prompts, MBPP: 500 problems; $\rho = 0.9$, 53-bit PRF). A window repeats at 16\% of positions on ELI5 and 34\% on MBPP. Detection at a false-positive rate of $10^{-3}$ (thresholds over 10 keys and 2048 / 1000 unwatermarked traces) and median z-score of the key-only test; both rules use the same detector. $\Delta$: paired difference against one use per window, with 95\% bootstrap interval over prompts. Distortion check, on 480 traces generated with one key per answer: the mean of $-\ln p_x - H(\vec p)$ at a window's 2nd / 3rd+ occurrences is $+3$ $[-6, +13]$ / $-3$ $[-8, +2]$ millinats with leftover continuation, and $-3$ $[-12, +7]$ / $+2$ $[-3, +6]$ millinats without any watermark, so leftover continuation introduces no measurable distortion.}
\label{tab:repeat}
\setlength{\tabcolsep}{4pt}
\begin{tabular}{lcccc}
\toprule
& \multicolumn{2}{c}{ELI5} & \multicolumn{2}{c}{MBPP} \\
Repeated windows & Detected @256 & Median z @256 / 1024 / 2048 & Detected @512 & Median z @1024 / 2048 \\
\midrule
One use per window & 93.0 & 6.0 / 15.9 / 22.8 & 53.2 & 4.8 / 5.6 \\
Leftover continuation & \textbf{96.8} & \textbf{6.5} / \textbf{16.9} / \textbf{23.9} & \textbf{60.8} & \textbf{5.2} / \textbf{5.7} \\
\midrule
$\Delta$ leftover continuation & $+3.8$ {\tiny$[+2.1, +5.6]$} & $+0.5$ / $+0.9$ / $+1.1$ & $+7.6$ {\tiny$[+2.4, +12.8]$} & $+0.3$ / $+0.4$ \\
\bottomrule
\end{tabular}
\end{table}

\section{Experimental details}
\label{app:details}

\revdel{This appendix gives the settings of every experiment: generation, speculative decoding, detection, diversity, distortion tests, compute, and our reimplementation of textGrain.
}
\paragraph{Generation.}
Prompts are \revdel{the first rows of the} ELI5 \revdel{dataset,} \rev{questions (a version without math questions, 267{,}901 rows), taken in order and} formatted with each model's chat template with thinking disabled; decoding uses temperature 1, top-$p$ 0.95, and at most 200 new tokens, with watermark window size $k = 3$. The copulas (\tsii and AutoMark's) use the 53-bit PRF of \autoref{app:prf} in \autoref{tab:main}, Figures~\ref{fig:teaser}, \ref{fig:three_models}, \ref{fig:pareto200} and~\ref{fig:distinct4}, and Tables~\ref{tab:entropy}, \ref{tab:code} and~\ref{tab:repeat}, as well as in every speculative-decoding experiment on Qwen3.5-9B; the remaining experiments (robustness, four keys, and the Llama drafter--target pair of \autoref{tab:llamapair}) use \ts's 13-bit PRF, which slightly understates \tsii (\autoref{app:prf}).
Detection runs use 3000 answers per configuration on \revdel{Qwen3.5-9B,} \rev{Qwen3.5-9B (rows 0--2999, one answer each),} with or without speculative decoding, and 1000 on the other \revdel{models;} \rev{models (rows 0--999);} null calibration uses 10{,}000 unwatermarked answers from the same model \revdel{(the} \rev{(rows 0--9999; rows 100{,}000--109{,}999 for the speculative-decoding experiments, with the} target \revdel{model under speculative decoding).} \rev{model).}
\revdel{Key values follow \ts's derivation: $k^{(1)} = 42$ and $k^{(2)} = 42 + 12345$; textGrain uses $B = m = 256$ and keys 42 and 4242. For null calibration, $10^5 j$ is added to every key for $j = 1, \dots, 9$, giving ten keys in total including the generation key.}
\rev{Key values follow \ts's derivation: $k^{(1)} = 42$ and $k^{(2)} = 42 + 12345$; textGrain uses $B = m = 256$ and keys 42 and 4242. For null calibration, adding $10^5 j$ to every key, $j = 0, \dots, 9$, gives ten keys, $j = 0$ being the generation key.}

\paragraph{Speculative decoding.}
In each round, Qwen3.5-0.8B proposes $N$ drafts, one at a time, from its own distribution (after temperature and top-$p$) through the scheme's sampler; Qwen3.5-9B then evaluates $\vec p_1, \dots, \vec p_{N+1}$ in a single forward pass, and the round emits the kept drafts and one more token, by block verification~\citep{sun2025block} with the closed forms of Appendix~\ref{app:proofs-block}.
A draft that is an end-of-sequence token ends the round's drafting, so the round then has fewer than $N$ drafts.
\tsii follows \autoref{alg:sampler}: drafts by the race \eqref{eq:draft}, then the residual race \eqref{eq:relay}, with the copula noise $\eta$ drawn once per position and read by both races.
The unwatermarked reference is native block verification with freshly drawn randomness.
$N = 8$ is the main setting; \autoref{tab:ndrafts} reports $N = 4$ and $16$, for which every configuration was generated again.
During generation, each answer tracks the (window, key) pairs whose values have been used, and every pair that a round reads counts as used, \revdel{whether} \rev{kept} or \revdel{not} \rev{not,} \revdel{its} \rev{for} \revdel{position} \rev{every} \revdel{is} \rev{scheme} \revdel{kept (Appendix~\ref{app:proofs-windows}); this rule applies to all schemes.} \rev{(Appendix~\ref{app:proofs-windows}).}
When a window has already been used, \ts and SynthID-Text fall back to their other key and then to fresh randomness, and textGrain and the copulas directly to fresh randomness (for \tsii, freshly drawn lane lengths, on which the draft race and the residual race still run).
Both models keep a key--value cache that is rolled back to the kept prefix after each round; on 40 random rounds, the logits after a rollback match a recomputation without cache as closely as those of a cache without rollback (bfloat16 noise).

\paragraph{Detection.}
Under speculative decoding with 8 drafts, the last token of each round makes up 36--37\% of the text (first 1000 answers at $\alpha = 0.5$ and $\rho = 0.9$). Scored under the draft key, AutoMark's last tokens, drawn under an independent key except at the bonus position (5\% of rounds), average 0.84 \revdel{-- below} \rev{(below} the null mean of 1, because a token drawn from block verification's residual law lost its draft \revdel{race --} \rev{race),} against 1.42 for its kept drafts; \tsii's last tokens average 1.34 under $s^\star$, against 1.42 for its kept drafts (all $\pm 0.004$). Positions whose window was already used get fresh values: 9\% of \tsii's tokens, against 6\% of AutoMark's.
For every detector, each (window, token) pair is scored once, at its first occurrence.
\revdel{\emph{Entropy-weighted test.} Each scored position $t$ is assigned the entropy $e_t$ of the next-token distribution of an auxiliary model -- here the generating model itself (the target model under speculative decoding), evaluated on the answer under examination -- and the weight $w_t = 0.1 + 0.9 \min\big(1, \max(0, (e_t - e_{\min})/(e_{\max} - e_{\min}))\big)$, where $e_{\min}$ and $e_{\max}$ range over the scored positions of the text. With per-position scores $s_t$ of null mean 1 and null variance $\theta$ -- $\theta = 1$ for single-key scores, and $\theta = a^2 + (1-a)^2$ for two-key fused scores $a\,s^{(1)} + (1-a)\,s^{(2)}$ -- the weighted sum $\sum_t w_t s_t$ is compared with the Gamma distribution of the same mean $\mu = \sum_t w_t$ and variance $\sigma^2 = \theta \sum_t w_t^2$ (shape $\mu^2/\sigma^2$, scale $\sigma^2/\mu$). Since all decision thresholds are empirical (below), this moment matching affects power but not validity.}
\rev{\emph{Entropy-weighted test.} Each scored position $t$ is assigned the entropy $e_t$ of the next-token distribution of an auxiliary model (here the generating model itself, the target model under speculative decoding, evaluated on the answer under examination) and the weight $w_t = 0.1 + 0.9 \min\big(1, \max(0, (e_t - e_{\min})/(e_{\max} - e_{\min}))\big)$, where $e_{\min}$ and $e_{\max}$ range over the scored positions of the text. With per-position scores $s_t$ of null mean 1 and null variance $\theta$ ($\theta = 1$ for single-key scores, $\theta = a^2 + (1-a)^2$ for two-key fused scores $a\,s^{(1)} + (1-a)\,s^{(2)}$), the weighted sum $\sum_t w_t s_t$ is compared with the Gamma distribution of the same mean $\mu = \sum_t w_t$ and variance $\sigma^2 = \theta \sum_t w_t^2$ (shape $\mu^2/\sigma^2$, scale $\sigma^2/\mu$). Since all decision thresholds are empirical (below), this moment matching affects power but not validity.}
The threshold for each detector and answer length is the empirical $10^{-3}$ quantile of its p-values on the 10{,}000 unwatermarked answers scored under 10 keys. Keys are interchangeable on unwatermarked text, so this pools $10^5$ null p-values, about 100 of which fall below the threshold; a single key would give only about 10, and hence a much noisier threshold. At the nominal $10^{-3}$ level, realised false-positive rates are 0.08--0.13\% for the detectors of \autoref{fig:calib} and 0.07--0.15\% for those used under speculative decoding (two-key fusions with the measured shares, and the Llama pair), so the Gamma and binomial tail approximations are already accurate, and the empirical thresholds absorb the remaining deviation.
Under speculative decoding, \ts, textGrain, SynthID-Text and AutoMark's copula fuse both keys' scores, weighted by their measured shares of emitted tokens (as \ts does for its coin flip); textGrain's own recipe uses a single key; \tsii uses $s^\star$ (private) and $s^{(1)}$ (public) exactly as without speculative decoding (\autoref{subsec:copula}).

\begin{figure}[ht]
\centering
\includegraphics[width=\linewidth]{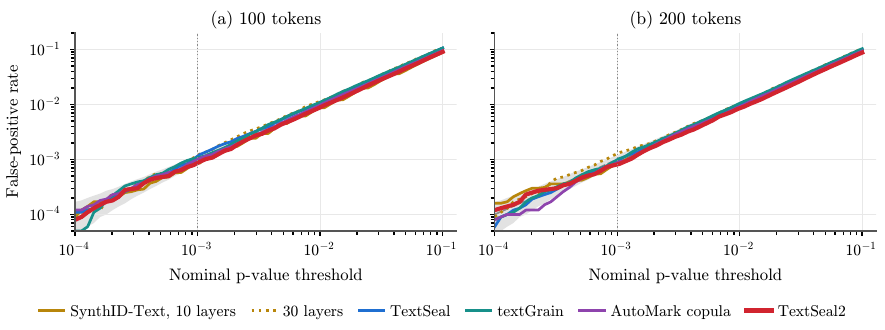}
\caption{Calibration: false-positive rate against the nominal p-value threshold for the detectors of \autoref{tab:main} without speculative decoding, on 10 keys and 10{,}000 unwatermarked Qwen3.5-9B answers; grey: 95\% binomial interval around the diagonal. At \revdel{thresholds of} $10^{-3}$ and above, every detector \revdel{stays} \rev{is} within a factor 1.3 of \revdel{the nominal rate.} \rev{nominal.}}
\label{fig:calib}
\end{figure}

\paragraph{Diversity.}
The Self-BLEU of a question is the mean sentence-level BLEU (sacrebleu) over all ordered pairs of its answers generated under the fixed key, decoded from their first 200 tokens. This is \ts's two-generation diversity estimator averaged over all answer pairs; we report the mean over questions. Intervals for values interpolated at \ts's Self-BLEU resample the 100 questions jointly for all methods (and detection outcomes binomially within questions), so that question-level variation common to all methods cancels in paired comparisons.
\autoref{fig:distinct4} repeats \autoref{fig:teaser} with diversity measured by the share of distinct 4-grams among a question's answers. The copulas still dominate the coin flip, and \tsii still leads under speculative decoding; without speculative decoding, textGrain is slightly ahead of the copulas under this metric (66\% vs.\ 62\% detected at \ts's distinct-4 level), whereas it is level with \tsii under Self-BLEU. \revdel{We use Self-BLEU as the primary metric for comparability with \ts.}

\begin{figure}[ht]
\centering
\includegraphics[width=\linewidth]{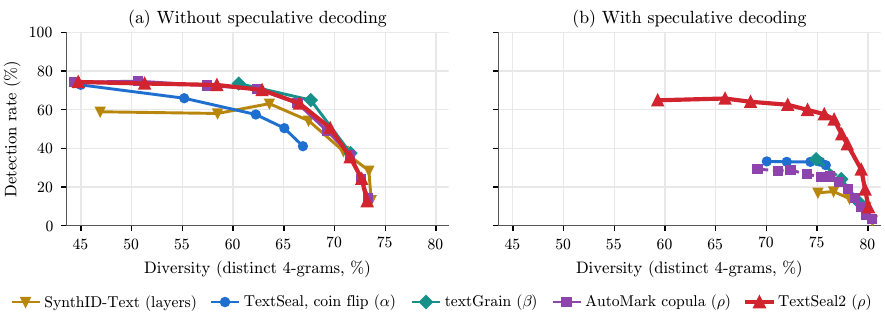}
\caption{\autoref{fig:teaser} with diversity measured by the share of distinct 4-grams (higher is more diverse), including the speculative panel (b).}
\label{fig:distinct4}
\end{figure}

\paragraph{Distortion tests.}
Answers used for acceptance and distortion measurements are generated with a fresh key $10^7 + 7919\,i$ for prompt $i$ (textGrain: this key and the key plus 4200), so that answers are independent draws \revdel{when averaged} over \revdel{keys --} \rev{keys,} the regime in which distortion-freeness is defined; 480 prompts per scheme and draft length.
The tests read the laws used at generation: for every emitted token, $p_x$ and $H(\vec p)$ come from the target pass of its round, and for the last token of each round, $\vec r_I$ is block verification's residual law at its position, computed from these $\vec p$ and $\vec q$.
All intervals are 95\% bootstrap intervals over answers.
\autoref{fig:ablation} reports the acceptance and distortion measurements for $N = 8$, and \autoref{tab:ndrafts} summarises them for $N = 4$ and $16$.
As a simulation control, the decisions were run on fixed toy laws with fresh values for every round ($N = 3$ drafts). The outcome $(d, I, y)$ of the residual race follows block verification's law ($\chi^2$ test over $6 \cdot 10^4$ rounds and 5 tokens, $p = 0.15$), whereas racing on the raw lane lengths (\autoref{rem:naive}) does not ($\chi^2 = 3194$ on 401 degrees of freedom). For the baselines ($6 \cdot 10^3$ rounds over 12 tokens), the first emitted token follows the target's law and the number of kept drafts follows native block verification ($p \ge 0.13$), except that textGrain's own recipe keeps fewer drafts, as it verifies the watermarked laws.

\begin{figure}[ht]
\centering
\includegraphics[width=0.88\linewidth]{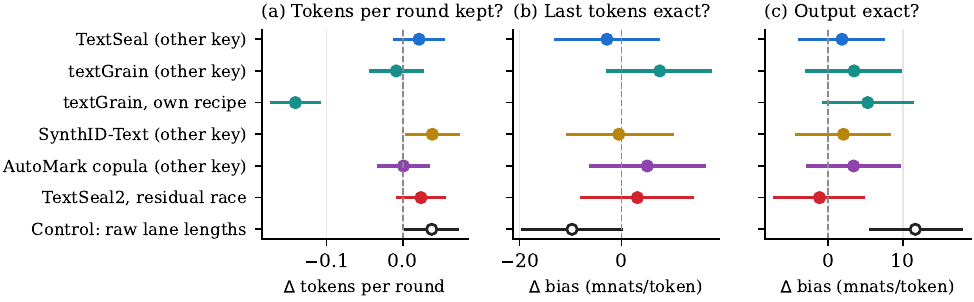}
\caption{Is speculative decoding with a watermark still distortion-free? Paired differences from unwatermarked block verification (8 drafts per round) on the same 480 prompts, with one key per answer (95\% paired bootstrap intervals). (a) Tokens per round. (b) Mean of $-\ln r_{I,y} - H(\vec r_I)$ at the last token $y$ of each round, which has expectation zero when it follows block verification's residual law $\vec r_I$ (not shown for textGrain's own recipe, whose last token follows the residual of the watermarked distributions). (c) Mean of $-\ln p_x - H(\vec p)$ over all tokens, which has expectation zero for any distortion-free scheme. Only the control \revdel{-- the residual race} \rev{(racing} on \revdel{the} raw lane \revdel{lengths} \rev{lengths,} \revdel{instead of the leftovers (\autoref{rem:naive}) --} \rev{\autoref{rem:naive})} deviates \revdel{detectably} in (c), and only textGrain's own recipe lowers \revdel{the} tokens per round.}
\label{fig:ablation}
\end{figure}

\paragraph{Compute.}
Each generation job uses one H200 GPU. Our speculative-decoding implementation (Hugging Face models with a key--value cache) is unbatched and unoptimised: with $N = 8$ it takes about 0.11\,s per token on one GPU, three times slower than plain decoding with the target, because the drafter is called once per draft; this affects only experimental throughput, not any reported rate, since block verification is exact. Under it, the residual race's decision takes 0.33\,ms per token, against 0.38\,ms for native block verification, and computing the keyed lane lengths of the drafted positions takes 2.4\,ms per token more than sampling the drafts without a key; this key-only work can run concurrently with the drafter's forward pass (\autoref{app:latency}).

\paragraph{textGrain reimplementation.}
The textGrain report~\citep{li2026textgrain} releases no code, so we implemented its Algorithms 1 and A.1 from the text: a keyed partition of the vocabulary into $B = 256$ blocks derived from the key and the window; a $256 \times 256$ table of standardised keyed Gumbel costs with $m = 256$ columns\rev{ (the report fixes neither $B$ nor $m$)}; the log-domain Sinkhorn updates with marginal correction and the regularisation update of its Eq.~(A.4); selection of a keyed column; and independent sampling of the block and then of the token within it. Detection follows its Section 3.2 ($\mathrm{Gamma}(n,1)$ null), except \revdel{that --} \rev{that,} as for every method in this \revdel{paper --} \rev{paper,} we score the first occurrence of each (window, token) pair rather than each window's first occurrence, which adds about 1\% more scored positions.
We validated the GPU kernels against a direct (unoptimised) implementation of the algorithm on random inputs, and verified on toy distributions that the emitted-token distribution is unbiased and that the entropy removed respects the budget $\beta$.
The speculative recipe implements its Appendix B with block verification between the watermarked distributions. OpenAI's implementation and hyperparameters may differ from ours, and our textGrain numbers carry that caveat.

\paragraph{Reproducibility.}
Generation and evaluation scripts, including fixed seeds, key derivations, prompt lists (ELI5 rows; MBPP problems) and threshold calibration, together with the per-answer outputs underlying every reported number, will be released with the paper.

\paragraph{Broader impact.}
Watermarks give providers, platforms and the public one signal for tracing AI-generated text. They are not robust to paraphrasing (\autoref{tab:robust}), can in principle be spoofed or removed by determined adversaries, and should therefore be combined with complementary provenance signals rather than treated as proof of origin.
Publishing a detection key enables independent public verification, but also gives adversaries a target signal to attack; retaining a private key allows such manipulation to be detected~\citep{milis2026public}.

\end{document}